\documentclass{article}

    \usepackage{amsthm,amsmath,amssymb,amsfonts, mathtools,ascmac,subcaption,comment,ulem,geometry,authblk,xcolor,caption,bbm,bm}
\usepackage{graphicx}

\pdfoutput=1

    \numberwithin{equation}{section}
    \theoremstyle{definition}
    
    \newtheorem{dfn}{Definition}[section]
    \newtheorem{prop}[dfn]{Proposition}
    \newtheorem{lem}[dfn]{Lemma}
    \newtheorem{thm}[dfn]{Theorem}
    \newtheorem{cor}[dfn]{Corollary}
    \newtheorem{rem}{Remark}

    \newcommand{\ve}{\varepsilon}

    \newcommand{\argmin}{\mathop{\rm arg~min}\limits}

\title{Estimating the Shannon Entropy Using the Pitman--Yor Process}
\author[1]{Takato Hashino}
\author[2]{Koji Tsukuda}
\affil[1]{Joint Graduate School of Mathematics for Innovation, Kyushu University}
\affil[2]{Faculty of Mathematics, Kyushu University.}
\date{}
\begin{document}
\maketitle

\begin{abstract}
    The Shannon entropy is a fundamental measure for quantifying diversity and model complexity in fields such as information theory, ecology, and genetics.
    However, many existing studies assume that the number of species is known, an assumption that is often unrealistic in practice.
    In recent years, efforts have been made to relax this restriction.
    Motivated by these developments, this study proposes an entropy estimation method based on the Pitman--Yor process, a representative approach in Bayesian nonparametrics.
    By approximating the true distribution as an infinite-dimensional process, the proposed method enables stable estimation even when the number of observed species is smaller than the true number of species.
    This approach provides a principled way to deal with the uncertainty in species diversity and enhances the reliability and robustness of entropy-based diversity assessment.
    In addition, we investigate the convergence property of the Shannon entropy for regularly varying distributions and use this result to establish the consistency of the proposed estimator.
    Finally, we demonstrate the effectiveness of the proposed method through numerical experiments.
\end{abstract}

\vspace{2.5truemm}


\section{Introduction}\label{sec:intro}
Entropy is a measure used to characterize diversity and model complexity. 
Among various notions of entropy, the Shannon entropy, also known as the Shannon--Wiener index, is a fundamental quantity that has been extensively studied in statistics and machine learning.
The Shannon entropy is employed in various fields such as information theory, ecology, genetics, and natural language processing.
For diversity indices, see, for example, Leinster~\cite{Leinster} and Patil and Taillie~\cite{Patil_1982}.
The problem of estimating the Shannon entropy of a population distribution has been studied for a long time \cite{Archer2014,Arora,ChaoShen_2003,Condit_2002,Hausser}.
In particular, we focus on the construction of an accurate estimator that performs well in small-sample regimes and remains applicable when the true number of species is unknown, as considered in Archer et al.~\cite{Archer2014}.

Throughout this paper, we use the following notation.
For $M \in \mathbb{N}$, let
$\Delta_M:=\{\boldsymbol{p}\in\mathbb{R}^M|\sum_{i=1}^M p_i=1, \ p_i>0 \ (i=1,\ldots,M)\}$ and let $\Delta_\infty := \{(p_1,p_2,\ldots) | \sum_{i=1}^\infty p_i=1, \ p_i>0 \ (i=1,2,\ldots)\}$.
Let $K \in \mathbb{N}$ denote the number of species, and consider a $K$-dimensional probability vector $\bm{p}=(p_1,\dots,p_K)\in\Delta_{K}$.
The Shannon entropy of $\bm{p}$, denoted by $H = H(\bm{p})$, is defined as
\begin{equation}
    H(\bm{p})=-\sum_{i=1}^Kp_i\log p_i\label{def:shannon}.
\end{equation}
We consider the problem of estimating $H(\boldsymbol{p})$ based on an independent and identically distributed sample of size $N \in \mathbb{N}$ drawn from the population distribution $\boldsymbol{p}$.

In this study, we consider the problem of estimating $H$ in the setting where both $K$ and $\boldsymbol{p}$ are unknown, as in \cite{Archer2014, Jiao_2015}.
Let $T$ denote the number of observed species in the random sample and let $n_1, \ldots, n_T$ denote their observed frequencies.
When the sample size $N$ is sufficiently large relative to $K$, it is reasonable to assume that $T = K$.
In this case, the simplest and most widely used approach is to estimate $p_1, \ldots, p_T$ by maximum likelihood and substitute the resulting estimates into \eqref{def:shannon}.
Specifically, the maximum likelihood estimator of each probability mass is given by
\begin{equation*}
    \hat{p}^{\mathrm{ML}}_i=\frac{n_i}{N}, \quad i=1,\ldots,T
\end{equation*}
and the corresponding maximum likelihood estimator of the entropy, denoted by $\hat{H}^{\mathrm{ML}}$, is
\begin{equation*}
    \hat{H}^{\mathrm{ML}} = H((\hat{p}^{\mathrm{ML}}_1,\ldots,\hat{p}^{\mathrm{ML}}_{T})) =-\sum_{i=1}^{T} \frac{n_i}{N}\log\frac{n_i}{N}.
\end{equation*}
In the literature, alternative estimators have been proposed, including the Miller--Madow estimator~\cite{Miller_1955} and the Chao--Shen estimator{~\cite{ChaoShen_2003}.
However, it has been indicated that these methods tend to perform poorly when the sample size $N$ is not sufficiently large relative to the true number of species $K$ \cite{Paninski_2004,Yihong_2016}.
To address this issue, Hausser and Strimmer~\cite{Hausser} proposed a shrinkage estimator for the individual probability masses toward $1/K$ under the assumption that $K$ is known, and demonstrated that the resulting estimator exhibits favorable properties even for small sample sizes.
In contrast, we consider the more challenging setting in which the sample size $N$ is smaller than the true number of species $K$, and $K$ is unknown.
Under this regime, we propose a new estimator of $H$.
The proposed approach is based on the Pitman--Yor process, which is widely used in the Bayesian nonparametric framework.
By employing an infinite-dimensional model induced by the Pitman--Yor process, the proposed method accounts for the contribution of unobserved species and enables stable entropy estimation.
Although the model distribution differs from the true distribution $\bm{p}$, simulation results show that our method performs well in estimating $H$, particularly when $K$ is large.
Moreover, it  performs comparably to the aforementioned conventional estimators when $N$ is large.

The remainder of the paper is organized as follows.
In Section~\ref{sec:pitman--yor}, we derive the predictive distribution induced by the Pitman--Yor process and the associated Shannon entropy.
In Section~\ref{sec:mixture_model}, we introduce the Dirichlet--Pitman--Yor mixture model considered in this study and discuss how it is used to construct an estimator of the Shannon entropy.
In Section~\ref{sec:proposed_estimator}, we present the proposed estimator for the Shannon entropy.
In Section~\ref{sec:finiteness}, we derive conditions for the finiteness of the Shannon entropy for discrete regularly varying distributions.
In Section~\ref{sec:consistency}, we establish the consistency of the proposed estimator.
Finally, in Sections~\ref{sec:simulation} and~\ref{sec:application}, we evaluate the performance of the proposed method through numerical experiments and applications to real data.

\section{Pitman--Yor predictive distribution}\label{sec:pitman--yor}

\subsection{Pitman--Yor process}
To address the setting in which the true number of species $K$ is unknown, we adopt a Bayesian nonparametric framework.
This framework allows us to define probability distributions on infinite-dimensional spaces without fixing the dimension of the probability vector in advance.
A detailed discussion of this approach can be found in Pitman and Yor~\cite{Pitman_1997}.
For a general introduction to Bayesian nonparametrics, see, for example, Ghosal and van der Vaart~\cite{Ghosal_2017}.

In particular, we employ the model known as the Pitman--Yor process (hereafter abbreviated as PYP).
This model is characterized by a discount parameter $d \in [0,1)$ and a concentration parameter $\alpha > -d$, and is denoted by $\mathrm{PY}(d,\alpha)$ throughout this paper.
A well-known construction of a $\boldsymbol{\pi} = (\pi_1, \pi_2, \ldots) \in \Delta_\infty$ distributed according to the Pitman--Yor process is given by the stick-breaking process.
Let $\{V_i\}_{i \geq 1}$ be a sequence of independent random variables, and define the associated sequence $\{\pi_i\}_{i \ge 1}$ by
\begin{equation*}
V_i \sim \mathrm{Beta}(1 - d,\alpha + i d), \quad
\pi_i = V_i \prod_{j=1}^{i-1} (1 - V_j), \quad
i=1,2,\ldots .
\end{equation*}
Then $\boldsymbol{\pi}$ forms an infinite sequence of nonnegative random variables whose sum converges to one almost surely, and we write $\boldsymbol{\pi} \sim \mathrm{PY}(d,\alpha)$.
Hereafter, we use the notation $\boldsymbol{V}_{1:k} = (V_1,\ldots,V_k)$ for $k=1,2,\ldots$, and $\boldsymbol{V}_{1:\infty} = (V_1,V_2,\ldots)$.

\begin{rem}
Let us explain the roles of the two parameters in $\mathrm{PY}(d,\alpha)$
The concentration parameter $\alpha$ controls how much probability mass is allocated to the leading components of $\boldsymbol{\pi}$.
For small values of $\alpha$, the first few components carry most of the probability mass, whereas for larger $\alpha$ the distribution becomes more uniform.
In contrast, the discount parameter $d$ governs the tail behavior of $\boldsymbol{\pi}$.
Larger values of $d$ lead to heavier power-law tails \cite{Archer2014}.
\end{rem}

\subsection{Pitman--Yor predictive distribution for the multinomial model}
 
Let $d \in [0,1)$ and $\alpha > -d$.
Consider the hierarchical model for a discrete random variable $Z$ given by
\[ \mathsf{P} (Z=k|\bm{\pi})=\pi_k (k=1,2,\ldots), \quad 
\bm{\pi}\sim\mathrm{PY}(d, \alpha) .
\]
Then the following proposition holds for the marginal distribution of $Z$, denoted by $P_{d,\alpha}^{\mathrm{MPY}}(\cdot) = \mathsf{P}(Z = \cdot )$.
Hereafter, we refer to the probability distribution $P_{d,\alpha}^{\mathrm{MPY}}(\cdot)$ as the marginal Pitman--Yor process.

\begin{prop}
Let $\bm{\pi}\sim\mathrm{PY}(d, \alpha)$.
For $k=1,2,\ldots$, suppose that $P(Z=k|\bm{\pi})=\pi_k$.
Then the marginal distribution of $Z$, denoted by $P_{d,\alpha}^{\mathrm{MPY}}$, is given by
    \begin{align}
        P_{d,\alpha}^\mathrm{MPY}(k)
        &=\frac{1-d}{\alpha}\prod_{j=1}^k
        \frac{\alpha+(j-1)d}{\alpha+(j-1)d+1}.\label{eq:MPY}
    \end{align}
\end{prop}

\begin{proof}
We derive the marginal distribution of $Z$ under the model $\mathsf{P}(Z=k|\boldsymbol{V}_{1:\infty})=\pi_{k}$, $k=1,2,\ldots$, $\boldsymbol{\pi}\sim\mathrm{PY}(d,\alpha)$.
For $k=1,2,\ldots$, we have
\begin{align*}
    P_{d,\alpha}^\mathrm{MPY}(k)
    &=\int \mathsf{P}(Z=k|\mathbf{V}_{1:k}) \mathsf{P}(\mathrm{d} \mathbf{V}_{1:k}|d,\alpha)\\
    &=\mathsf{E}\left[V_{k} \prod_{j=1}^{k-1}(1-V_j)\right]\\
    &=\mathsf{E}[V_k]\prod_{j=1}^{k-1}\mathsf{E}\left[1-V_j\right] \\
    &=\frac{1-d}{\alpha+(k-1)d+1}\prod_{j=1}^{k-1}\frac{\alpha+jd}{\alpha+(j-1)d+1} \\
    &=\frac{1-d}{\alpha}\prod_{j=1}^{k}\frac{\alpha+(j-1)d}{\alpha+(j-1)d+1}.
\end{align*}
Here, we use the independence of $\{V_i\}_{i\geq1}$ and the expectation of the Beta distribution.
\end{proof}

\begin{rem}
When $d=0$, equation~\eqref{eq:MPY} reduces to
\begin{align*}
    P_{0,\alpha}^\mathrm{MPY}(k)
    =\frac{1}{\alpha+1}\left(1-\frac{1}{\alpha+1}\right)^{k-1}
\end{align*}
which corresponds to a geometric distribution with parameter $1/(\alpha+1)$.
For $d\neq0$, equation~\eqref{eq:MPY} can be rewritten using the Gamma function as 
\begin{equation}\label{eq:MPY2}
    P_{d,\alpha}^\mathrm{MPY}(k)=\frac{1-d}{\alpha}\cdot\frac{\Gamma(\frac{\alpha}{d}+k)}{\Gamma(\frac{\alpha}{d})}\cdot\frac{\Gamma(\frac{\alpha+1}{d})}{\Gamma(\frac{\alpha+1}{d}+k)}
\end{equation}
which corresponds to the Waring distribution.
See, for example, Chapter~11 of Johnson et al.~\cite{Jonson_2005} for further details on the Waring distribution.
\end{rem}

Figure~\ref{fig:prePY} shows the probability mass functions of the marginal Pitman--Yor process.
The upper-left and upper-right panels correspond to $\alpha = 49$ and $\alpha = 99$.
In the top row, the predictive distributions are shown for different values of $d$.
The gray curve corresponds to $d = 0$ and reduces to the geometric distribution.
The light blue and dark blue curves correspond to $d = 0.5$ and $d = 0.75$, respectively.
In the bottom row, the parameters $(d,\alpha)$ are chosen so that $P_{d,\alpha}^\mathrm{MPY}(1) =0.02$ for all cases.
Figure~\ref{fig:prePY} shows that the marginal Pitman--Yor process reflects the effects of the Pitman--Yor parameters.
In particular, the discount parameter $d$ controls the tail behavior, ranging from light-tailed to heavy-tailed distributions, while the concentration parameter $\alpha$ governs the overall shape of the distribution.

\begin{figure}[t]
    \centering
    \begin{minipage}{0.45\textwidth}
        \centering
        \includegraphics[width=\textwidth]{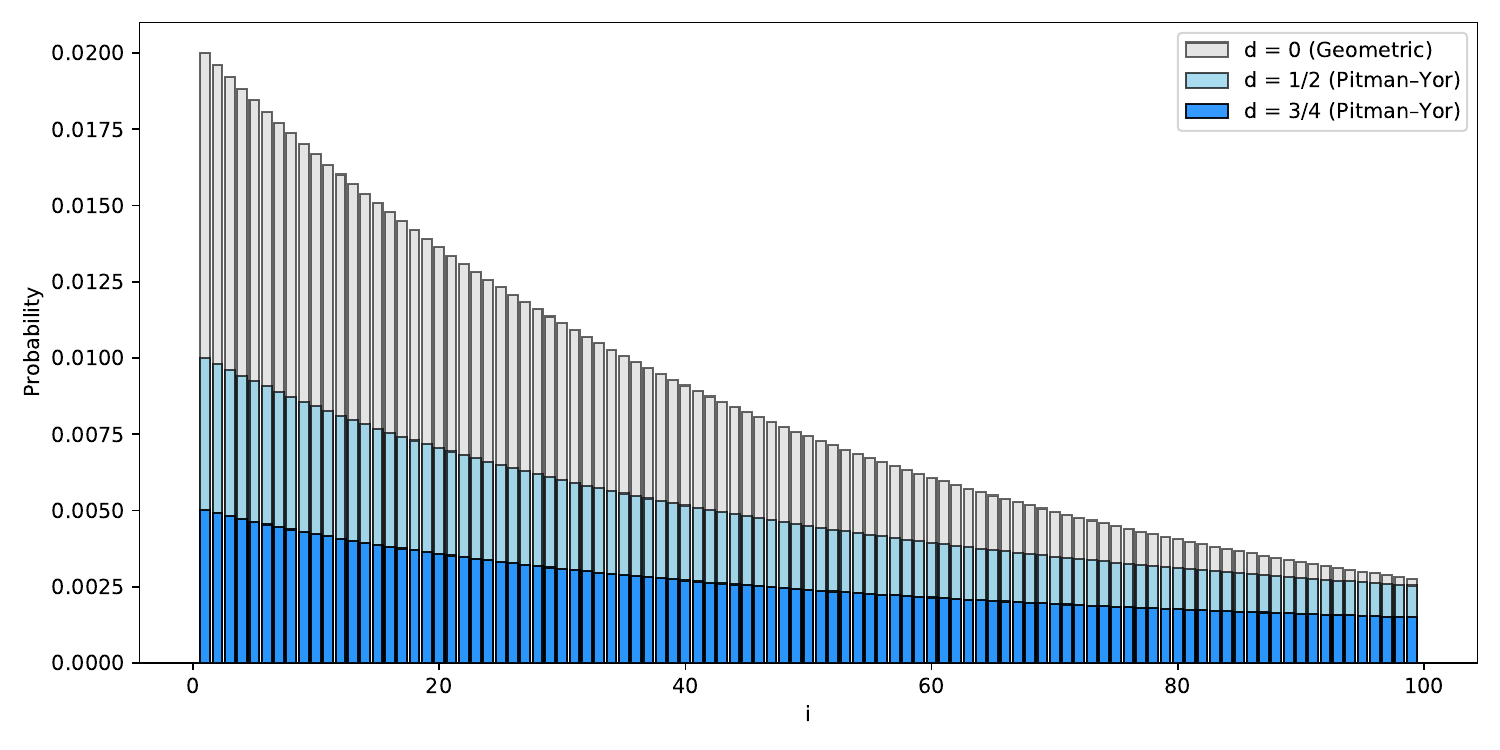}
    \end{minipage}%
    \hspace{1em}
    \begin{minipage}{0.45\textwidth}
        \centering
        \includegraphics[width=\textwidth]{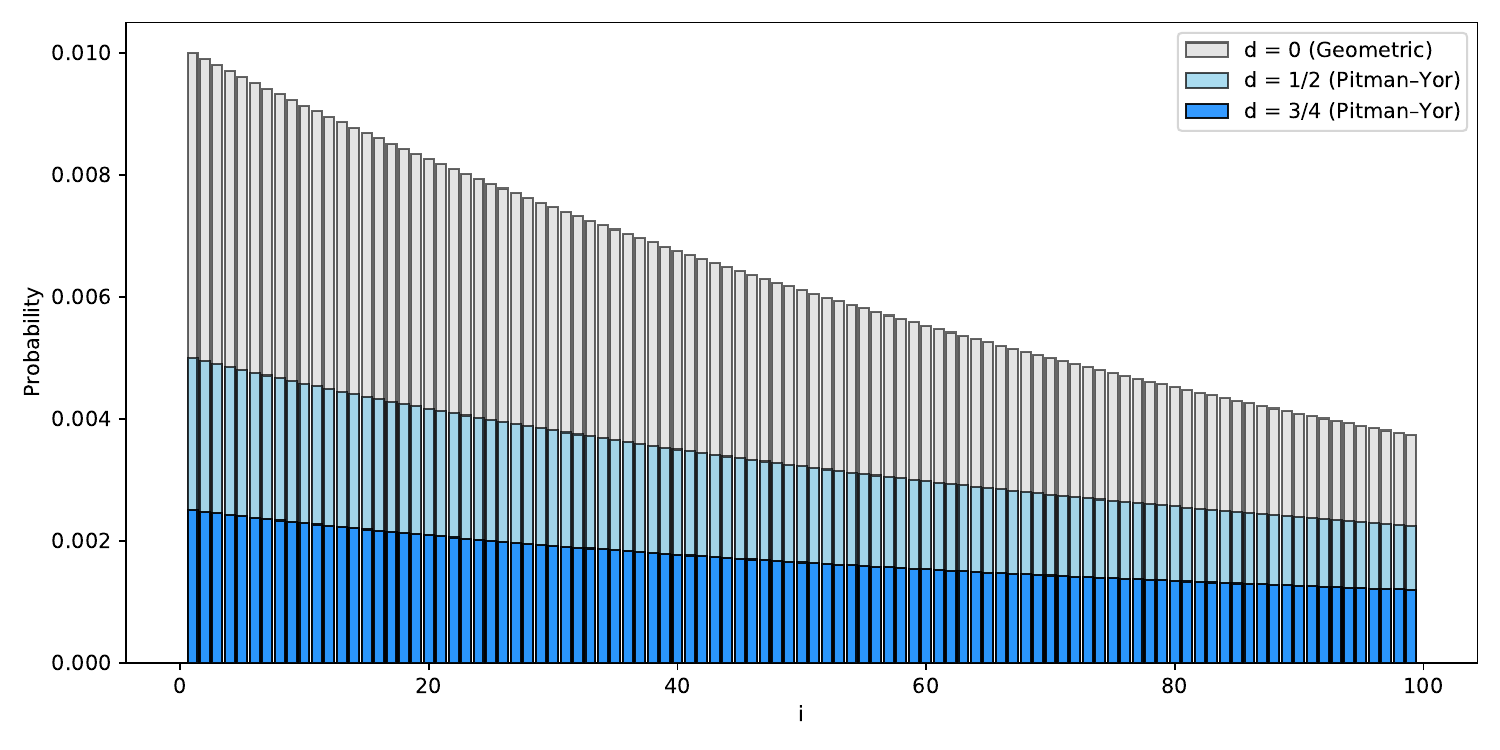}
    \end{minipage}%
    \hspace{1em}
    \begin{minipage}{0.45\textwidth}
        \centering
        \includegraphics[width=\textwidth]{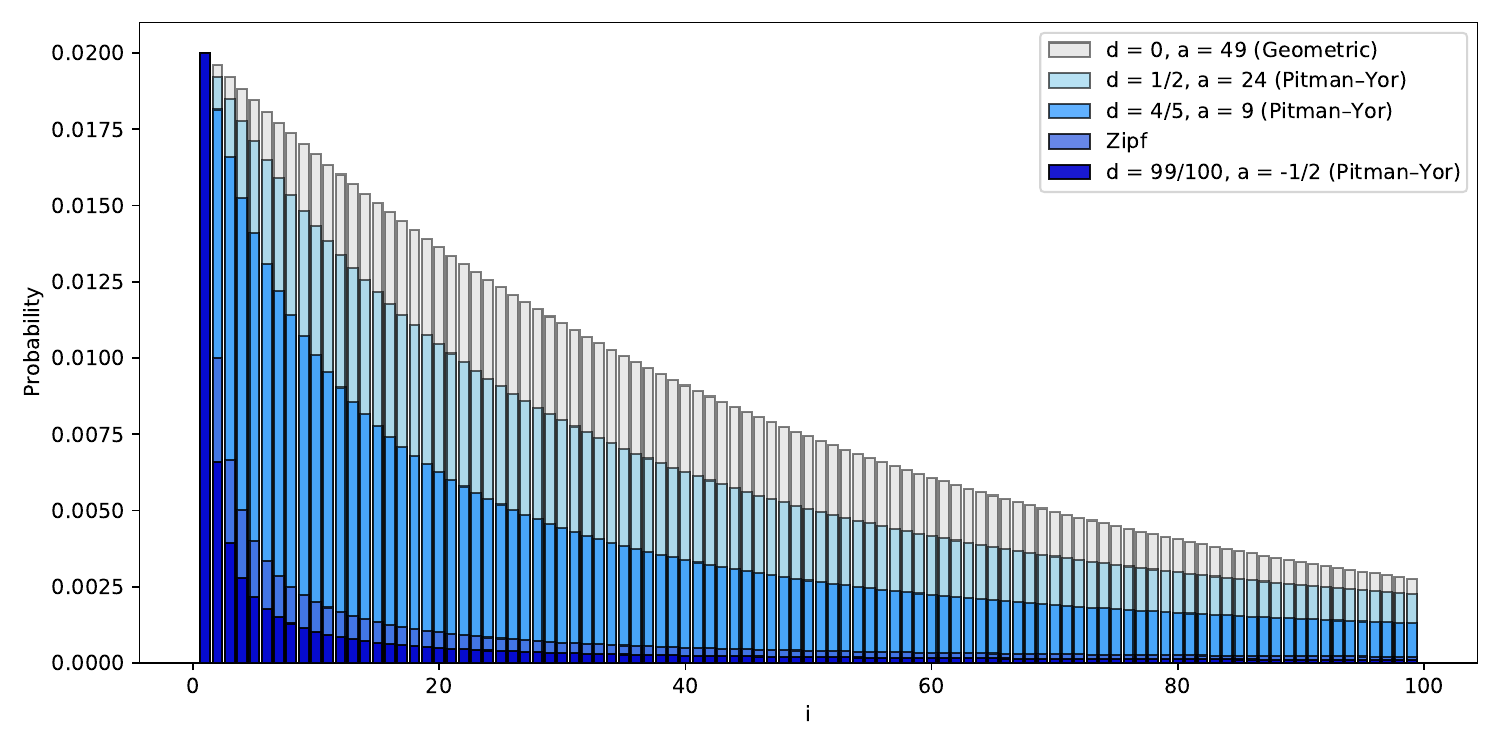}
    \end{minipage}%
    \caption{
    The probability mass functions of the marginal Pitman--Yor process.
    The upper-left and upper-right panels correspond to $\alpha=49$ and $\alpha=99$, respectively.
    The top row shows probability mass functions for different values of the discount parameter $d$: $d=0$ (gray, geometric), $d=0.5$ (light blue), and $d=0.75$ (dark blue).
    In the bottom row, the parameters $(d,\alpha)$ are chosen such that $P_{d,\alpha}^{\mathrm{MPY}}(1)=0.02$ in all cases.
    }
  \label{fig:prePY}
\end{figure}

\subsection{The Shannon entropy of MPY}
We consider the Shannon entropy
\begin{equation}
H(\bm{p}) = -\sum_{i=1}^\infty p_i\log p_i\label{def:shannon2}
\end{equation}
of an infinite-dimensional probability vector $\boldsymbol{p} = (p_1, p_2, \ldots) \in \Delta_\infty$ whose $k$-th component is given by \eqref{eq:MPY}.
When $d=0$, the probability mass function is given by 
\[ p_k = \frac{1}{\alpha+1}\left(1-\frac{1}{\alpha+1}\right)^{k-1} \quad (k=1,2,\dots),\]
and the Shannon entropy can be computed explicitly as
\begin{equation*}
    H(\bm{p})=(1+\alpha)\log(1+\alpha)-\alpha\log\alpha.
\end{equation*}
When $d \neq 0$, no closed-form expression for the entropy is available.
The Shannon entropy $H(\bm{p})$ is given by substituting \eqref{eq:MPY} into \eqref{def:shannon2}.
Since the support of the distribution is infinite, numerical computation requires approximation.
For numerical evaluation, we truncate the infinite sum at a sufficiently large index and approximate the remaining tail using a combination of Stirling's formula and integration by parts.
Regarding the approximation, the following proposition holds.

\begin{prop}
Let $\bm{p} \in \Delta_{\infty}$ be $\bm{p}=(P_{d,\alpha}^\mathrm{MPY}(1),P_{d,\alpha}^\mathrm{MPY}(2),\ldots)$ defined by \eqref{eq:MPY}.
If $d \neq 0$, then as $n \to \infty$, the Shannon entropy satisfies
\begin{align*}
    H(\bm{p})
    &=-\log\left(\frac{(1-d)\Gamma(\frac{\alpha+1}{d})}{\alpha\Gamma(\frac{\alpha}{d})}\right)-\frac{(1-d)\Gamma(\frac{\alpha+1}{d})}{\alpha\Gamma(\frac{\alpha}{d})}\sum_{k=1}^{n}\left(\frac{\Gamma(\frac{\alpha}{d}+k)}{\Gamma(\frac{\alpha+1}{d}+k)}\right)\log\left(\frac{\Gamma(\frac{\alpha}{d}+k)}{\Gamma(\frac{\alpha+1}{d}+k)}\right)\\
    &\quad + \frac{(1-d)\Gamma(\frac{\alpha+1}{d})}{\alpha\Gamma(\frac{\alpha}{d})} \left\{ \frac{(n+1)^{\frac{d-1}{d}}}{1-d}\log(n+1)+\frac{d(n+1)^{\frac{d-1}{d}}}{(d-1)^2} \right\} +O\left(n^{-\frac{1}{d}}\log n\right).
\end{align*}
\end{prop}

\begin{proof}
When $d\neq0$, $p_k$ is given by \eqref{eq:MPY2} for $k=1,2,\ldots$.
From \eqref{def:shannon2}, the Shannon entropy is given by
\begin{align*}
    H(\bm{p})
    &=-\sum_{k=1}^{\infty}p_k\log\left(\frac{1-d}{\alpha}\cdot\frac{\Gamma(\frac{\alpha}{d}+k)}{\Gamma(\frac{\alpha}{d})}\cdot\frac{\Gamma(\frac{\alpha+1}{d})}{\Gamma(\frac{\alpha+1}{d}+k)}\right)\\
    &=-\log\left(\frac{(1-d)\Gamma(\frac{\alpha+1}{d})}{\alpha\Gamma(\frac{\alpha}{d})}\right)-\frac{(1-d)\Gamma(\frac{\alpha+1}{d})}{\alpha\Gamma(\frac{\alpha}{d})}\sum_{k=1}^{\infty}\left(\frac{\Gamma(\frac{\alpha}{d}+k)}{\Gamma(\frac{\alpha+1}{d}+k)}\right)\log\left(\frac{\Gamma(\frac{\alpha}{d}+k)}{\Gamma(\frac{\alpha+1}{d}+k)}\right).
\end{align*}
We split the infinite sum appearing in the second term on the right-hand side into the partial sum up to the $n$-th term and the remaining tail, and evaluate the latter.
For the ratio of Gamma functions, it is well known that for $a,b>0$, as $x \to \infty$,
\begin{equation*}
    \frac{\Gamma(x+a)}{\Gamma(x+b)}=x^{a-b}\left(1+O\left(\frac{1}{x}\right)\right).
\end{equation*}
Moreover, for $a>1$, we have
\begin{equation*}
    \int_{n+1}^\infty x^{-a}\log(x) \mathrm{d}x
    = - \frac{(n+1)^{1-a}}{1-a}\log(n+1)+\frac{(n+1)^{1-a}}{(1-a)^2}.
\end{equation*}
Applying these relations, the tail sum is approximated as
\begin{align*}
    &\sum_{k=n+1}^{\infty}\left(\frac{\Gamma(\frac{\alpha}{d}+k)}{\Gamma(\frac{\alpha+1}{d}+k)}\right)\log\left(\frac{\Gamma(\frac{\alpha}{d}+k)}{\Gamma(\frac{\alpha+1}{d}+k)}\right)\\
    &=\sum_{k=n+1}^{\infty}k^{-\frac{1}{d}}\left(1+O\left(\frac{1}{k}\right)\right) \log \left(k^{-\frac{1}{d}}\left(1+O\left(\frac{1}{k}\right)\right)\right)\\
    &=\sum_{k=n+1}^\infty\left\{k^{-\frac{1}{d}} \log\left(k^{-\frac{1}{d}}\right) +O\left(k^{-(1+\frac{1}{d})}\log k\right)\right\}\\
    &={-\frac{1}{d}}\sum_{k=n+1}^\infty k^{-\frac{1}{d}}\log k+O\left((n+1)^{-\frac{1}{d}}\log(n+1)\right)\\
    &=-\frac{1}{d}\left(\int_{n+1}^\infty x^{-\frac{1}{d}}\log(x) \mathrm{d}x + O\left(n^{-\frac{1}{d}}\log n\right)\right)+O\left((n+1)^{-\frac{1}{d}}\log(n+1)\right)\\
    &=\left\{ \frac{(n+1)^{\frac{d-1}{d}}}{d-1}\log(n+1)-\frac{d(n+1)^{\frac{d-1}{d}}}{(d-1)^2}\right\} +O\left(n^{-\frac{1}{d}}\log n\right).
\end{align*}
This completes the proof.
\end{proof}

\section{Mixture model with a Pitman--Yor process}\label{sec:mixture_model}

\subsection{Dirichlet--Pitman--Yor mixture model}

Given a random sample of size $N$, let $T$ denote the number of observed species, and $n_i$ the number of observations of species $i$ for $i = 1,\ldots,T$.
We write $\mathbf{y}=(n_{1},\dots,n_{T})^{\top}$.
Then, we have $N = \sum_{i=1}^T n_i$ by definition. 
Archer et al.~\cite{Archer2014} introduced the following mixture model for constructing a probability vector based on the Pitman--Yor process.

\begin{dfn}(Dirichlet--Pitman--Yor mixture model)~\cite{Archer2014}
Given the observation $\mathbf{y}=(n_{1},\dots,n_{T})^{\top}$, we define a random infinite-dimensional probability vector $\bm{q}\in\Delta_{\infty}$ as follows:
\begin{align*}
& (q_{1},\dots,q_{T}, \tilde{q}) | \mathbf{y} \sim \mathrm{Dir}_{T+1}(n_{1}-d,\dots,n_{T}-d,\alpha+Td),\\
& \boldsymbol{\pi} | \mathbf{y}\sim \mathrm{PY}(d,\alpha+Td),\\
& \bm{q} = (q_{1},\dots,q_{T}, \tilde{q} \boldsymbol{\pi}).
\end{align*}
Here, $\mathrm{Dir}_{M}(\theta_{1},\ldots,\theta_{M})$ denotes the $M$-dimensional Dirichlet distribution with parameter vector $(\theta_{1},\ldots,\theta_{M})$.
We refer to this construction as the Dirichlet--Pitman--Yor mixture model (hereafter abbreviated as DPYM), and write $\bm{q}\sim\mathrm{DPYM}(\mathbf{y},d,\alpha)$.
\end{dfn}

\begin{rem}
In the DPYM, probabilities for previously observed species are assigned through a Dirichlet distribution whose parameters are weighted by the observed frequencies, whereas probabilities for unobserved species are modeled using a Pitman--Yor process.
This construction enables the DPYM to reflect the information contained in the observed data, while flexibly allowing for the existence of unseen species.
\end{rem}

Under the DPYM, given the observed frequency $\mathbf{y}=(n_{1},\dots,n_{T})^{\top}$ of $N$ individuals, the predictive distribution for a new observation $Z$ is given by
\begin{align}
    &\mathsf{P}(Z = t|\mathbf{y},d,\alpha)
    =\frac{n_{t}-d}{N+\alpha} \quad (t = 1 ,\ldots, T) \label{eq:DPYMP1}, \\
    &\mathsf{P}(Z = T+k|\mathbf{y},d,\alpha)
    =\frac{1-d}{N+\alpha}\prod_{j=1}^{k}\frac{\alpha+(T+j-1)d}{\alpha+(T+j-1)d+1} \quad (k=1,2,\ldots). \label{eq:DPYMP2}
\end{align}
Using \eqref{eq:DPYMP1} and \eqref{eq:DPYMP2}, we define
\begin{align}
    &\bm{q}^*=\left(\mathsf{P}(Z=1|\mathbf{y},d,\alpha),\dots,\mathsf{P}(Z=T|\mathbf{y},d,\alpha),\frac{\alpha+Td}{N+\alpha}\right) =\left( \frac{n_1 - d}{N+\alpha} ,\dots, \frac{n_T - d}{N+\alpha} , \frac{\alpha+Td}{N+\alpha} \right), \label{eq:qstar} \\
    &\pi_{k} = \frac{N+\alpha}{\alpha+Td} \mathsf{P}(Z = T+k|\mathbf{y},d,\alpha) = \frac{1-d}{\alpha+Td}\prod_{j=1}^{k}\frac{\alpha+(T+j-1)d}{\alpha+(T+j-1)d+1} \quad (k=1,2,\ldots). \nonumber
\end{align}
We write $\boldsymbol{\pi} = (\pi_1,\pi_2,\ldots)$.
Then, an infinite dimensional probability vector $\bm{q}$ is given by
\begin{equation}
    \bm{q}=\left( \frac{n_1 - d}{N+\alpha} ,\dots, \frac{n_T - d}{N+\alpha} , \frac{\alpha+Td}{N+\alpha} \boldsymbol{\pi}\right). \label{defq}
\end{equation}
The Shannon entropy of $\bm{q}$ can be decomposed as
\begin{equation*}
    H(\bm{q})=H(\bm{q}^*)+\frac{\alpha+Td}{N+\alpha}H(\boldsymbol{\pi}).
\end{equation*}
Based on this model, estimation of the population entropy $H(\bm{p})$ depends on the choice of the hyperparameters $(d,\alpha)$.
Accordingly, we propose a method for selecting $(d,\alpha)$ based on minimizing the cross entropy between $\bm{q}$ and $\bm{p}$, as described in the next section.

\begin{rem}
The aim of this study is to provide a principled framework for bias correction in entropy estimation, rather than to identify the true data-generating model.
For this reason, we do not attempt to estimate the parameters $(d,\alpha)$ via model-fitting methods such as maximum likelihood estimation.
\end{rem}

\subsection{Data-driven parameter selection}
We consider a $K$-dimensional probability vector $\bm{p}= (p_1,\ldots,p_K) \in\Delta_K$, where both $K$ and $\bm{p}$ are unknown.
Given a random sample drawn from the underlying population, we approximate the true distribution $\bm{p}$ by the probability vector $\bm{q} \in \Delta_{\infty}$ defined by \eqref{defq}.
As a measure of discrepancy between two probability vectors $\bm{p}$ and $\bm{q}$, we use the cross entropy
\begin{equation}
    H(\bm{p},\bm{q})=-\sum_{i=1}^{K}p_i\log q_i.
\end{equation}
Moreover, $H(\bm{p},\bm{q})$ can be expressed by using the Kullback--Leibler divergence $D_{KL}(\bm{p} || \bm{q})$ and the Shannon entropy $H(\bm{p})$ of $\bm{p}$ as
\begin{equation*}
    H(\bm{p},\bm{q})=D_{KL}(\bm{p}||\bm{q})+H(\bm{p}).
\end{equation*}
By Jensen's inequality, we have $D_{KL}(\bm{p}||\bm{q}) \ge 0$ and $H(\bm{p}) \ge 0$.
Therefore, a smaller cross entropy  $H(\bm{p}, \bm{q})$ indicates a smaller discrepancy between $\bm{p}$ and $\bm{q}$.
In other words, by selecting $(d,\alpha)$ so as to minimize the cross entropy, we may expect $\bm{q}$ to provide a good approximation to $\bm{p}$.
The following proposition concerns the cross entropy between the true probability vector $\bm{p}$ and $\bm{q}$.

\begin{prop}\label{prop:upper}
Let $\bm{p}$ be a $K$-dimensional probability vector, and suppose that a random sample of size $N$ is drawn from $\bm{p}$, in which $T$ distinct species are observed.
Let $\mathbf{y}=(n_1,\ldots,n_T)$ denote the corresponding observation counts.
Then, the cross entropy $H(\bm{p}, \bm{q})$ between $\bm{p}$ and $\bm{q}$ defined by \eqref{defq} satisfies
\begin{equation}
   H(\bm{p}, \bm{q})\le \log(N+\alpha)-(C_{1}+C_{0})\log(1-d)+\frac{C_{0}}{2}(K-T+1) \log\left( \frac{\alpha+Td+1}{\alpha+Td} \right), \label{ueq:upper}
\end{equation}
where $C_1$ denotes the total probability mass of species observed exactly once, and $C_0$ denotes the total probability mass of unobserved species, as defined by
\begin{equation*}
    C_1=\sum_{i=1}^Kp_iI\{n_i=1\}, \quad C_0=1-\sum_{i=1}^Kp_iI\{n_i>0\},
\end{equation*}
where $I\{\cdot\}$ denotes the indicator function.
\end{prop}

\begin{proof}
It holds that
    \allowdisplaybreaks
    \begin{align*}
    &H(\bm{p},\bm{q})\\
    &=-\sum_{i=1}^{K}p_{i}\log q_{i}\\
    &=-\left\{\sum_{i=1}^{T}p_{i}\log\left(\frac{n_{i}-d}{N+\alpha}\right)+\sum_{i=1}^{K-T}p_{T+i}\log\left(\frac{1-d}{N+\alpha}\prod_{j=1}^{i}\frac{\alpha+(T+j-1)d}{\alpha+(T+j-1)d+1} \right) \right\}\\
    &=-\left\{\sum_{i=1}^{T}p_{i}\log(n_{i}-d)+\sum_{i=T+1}^{K}p_{i}\log (1-d)-\log(N+\alpha)\right\}+\sum_{i=1}^{K-T}p_{T+i}\sum_{j=1}^{i} \log \left( \frac{\alpha+(T+j-1)d+1}{\alpha+(T+j-1)d} \right) \\
    &\le\log(N+\alpha)-\sum_{i=1}^{T}p_{i}\log(n_{i}-d)-\sum_{i=T+1}^{K}p_{i}\log (1-d)+\sum_{i=1}^{K-T}p_{T+i}\sum_{j=1}^{i} \log\left(\frac{\alpha+Td+1}{\alpha+Td}\right) \\
    &=\log(N+\alpha)-\sum_{i=1}^{T}p_{i}\log(n_{i}-d)-\sum_{i=T+1}^{K}p_{i}\log (1-d)+\log \left( \frac{\alpha+Td+1}{\alpha+Td}\right) \sum_{i=1}^{K-T}ip_{T+i}\\
    &\le\log(N+\alpha)-(C_{1}+C_{0})\log(1-d) + \log\left(\frac{\alpha+Td+1}{\alpha+Td}\right) \sum_{i=1}^{K-T}ip_{T+i}\\
    &\le\log(N+\alpha)-(C_{1}+C_{0})\log(1-d)+\frac{C_{0}}{2}(K-T+1) \log \left( \frac{\alpha+Td+1}{\alpha+Td}\right).
\end{align*}
Here, the last inequality follows from Lemma~\ref{lem34} below.
\end{proof}

\begin{lem}\label{lem34}
Let $\tilde\Delta = \{ (p_1,\ldots,p_K) \in\mathbb{R}^K|\sum_{i=1}^Kp_i=1, p_i\ge 0 \ (i=1,\ldots,K) \}$.
Then,
\begin{equation}
    \max_{\boldsymbol{p} \in \tilde\Delta}\left\{\sum_{i=1}^Kip_i\mid p_1\ge p_2\ge\cdots\ge p_K,  \right\}=\frac{K+1}{2}.
\end{equation}
\end{lem}

\begin{proof}
    When $p_1=p_2=\cdots=p_K={1}/{K}$, we have
    \begin{equation*}
        \sum_{i=1}^Kip_i=\frac{K+1}{2}.
    \end{equation*}
    Moreover, by Chebyshev's sum inequality, it follows that
    \begin{align*}
        \sum_{i=1}^Kip_i
        &\le\left(\frac{1}{K}\sum_{i=1}^Ki\right)\left(\sum_{i=1}^Kp_i\right)\\
        &=\frac{K+1}{2}
    \end{align*}
    This completes the proof.
\end{proof}

Hereafter, we regard the right-hand side of \eqref{ueq:upper} as a function of the hyperparameters $(d,\alpha)$ and refer to it as the upper bound function, denoted by $f(d,\alpha)$.
By Proposition~\ref{prop:upper}, the cross entropy between the true probability vector $\bm{p}$ and its approximation $\bm{q}$ can be explicitly bounded by $f(d,\alpha)$, providing practical guidance for the selection of the hyperparameters $(d,\alpha)$.
Importantly, the bound holds without imposing any additional structural assumptions on the true distribution $\bm{p}$, other than the finiteness of its dimension.

To evaluate the upper bound function in concrete settings, we fix the discount parameter $d$ and examine the behavior of the Kullback--Leibler divergence $D_{KL}(\bm{p}||\bm{q})$ and the difference between the cross-entropy upper bound and the true entropy, $f(d,\alpha)-H(\bm{p})$, as functions of $\alpha \in [0,1000)$.
These quantities are shown in Figure~\ref{fig:KL}.
The left column corresponds to the case $d=0$, while the right column corresponds to $d=0.5$.
In the top row, the true probability vector $\bm{p}$ is drawn from a $\mathrm{Dir}_{500}(0.1)$, whereas in the bottom row, $\bm{p}$ is drawn from a $\mathrm{Dir}_{500}(10)$.
Here, $\mathrm{Dir}_M(\alpha)$ denotes the $M$-dimensional symmetric Dirichlet distribution with the parameter $\alpha > 0$}.

\begin{figure}[tbp]
    \centering
    \begin{minipage}{0.45\textwidth}
        \centering
        \includegraphics[width=\textwidth]{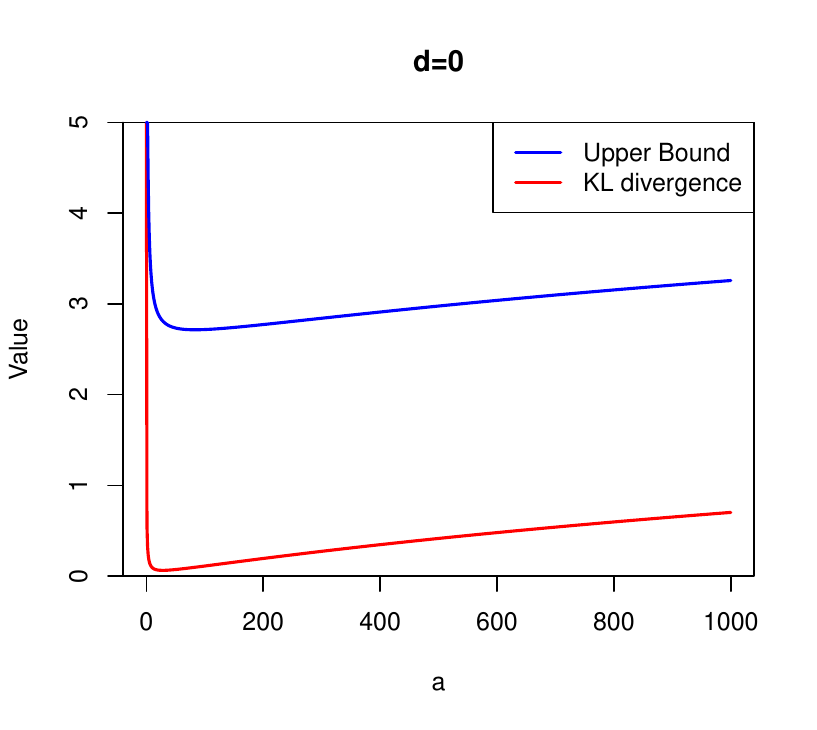}
    \end{minipage}%
    \hspace{1em}
    \begin{minipage}{0.45\textwidth}
        \centering
        \includegraphics[width=\textwidth]{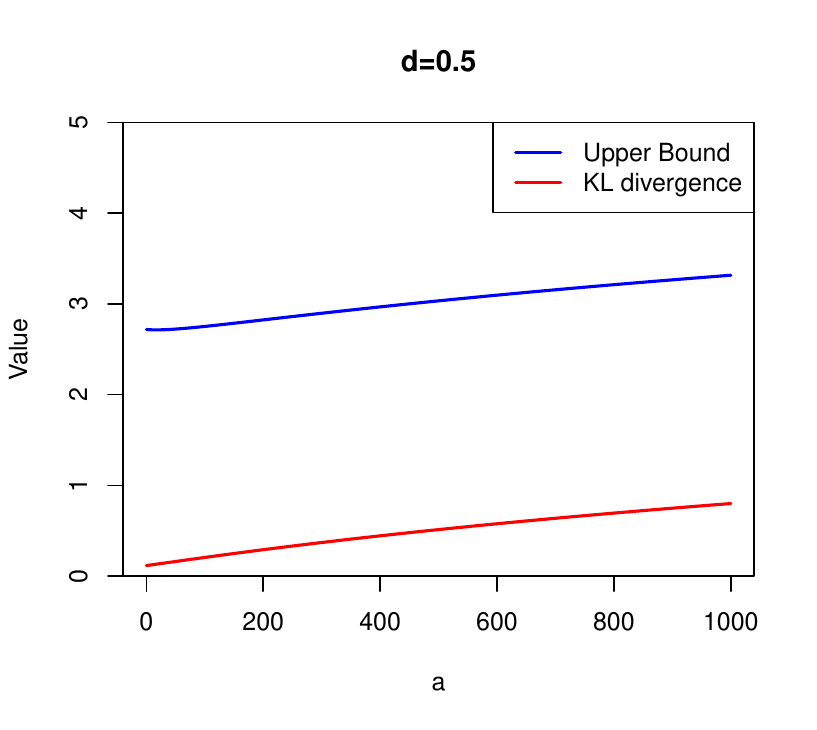}
    \end{minipage}%
    \vspace{1em}
    \begin{minipage}{0.45\textwidth}
        \centering
        \includegraphics[width=\textwidth]{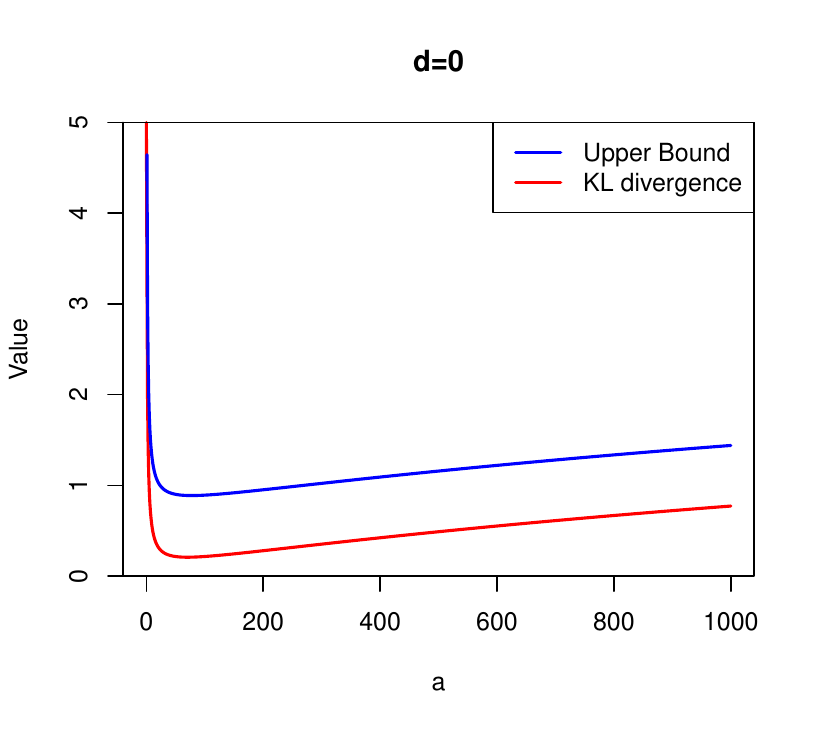}
    \end{minipage}%
    \hspace{1em}
    \begin{minipage}{0.45\textwidth}
        \centering
        \includegraphics[width=\textwidth]{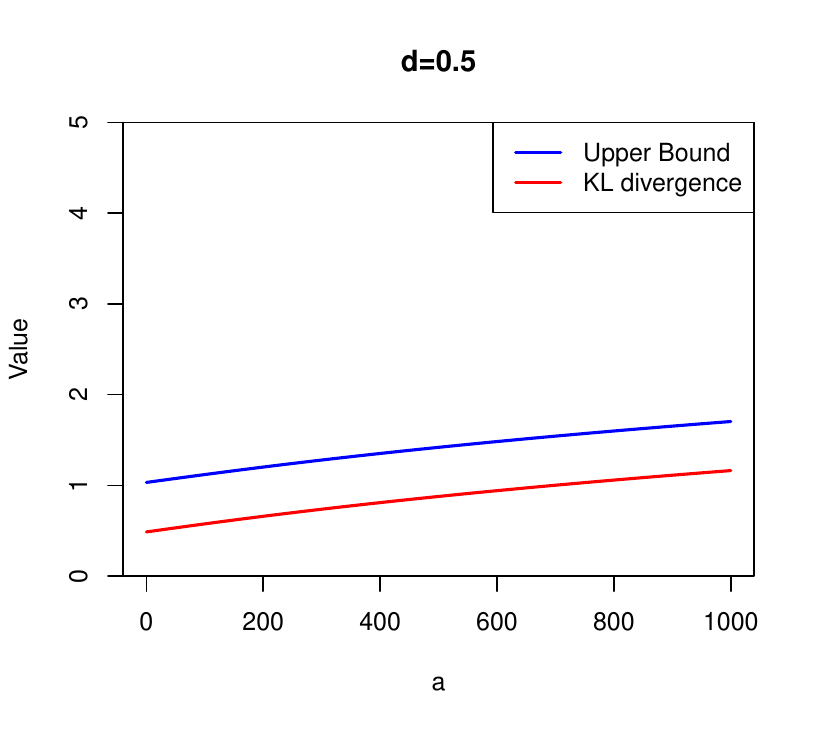}
    \end{minipage}%
    \caption{Behavior of the Kullback--Leibler divergence $D_{KL}(\bm{p}||\bm{q})$ and the difference between the cross-entropy upper bound and the true entropy, $f(d,\alpha)-H(\bm{p})$, as functions of $\alpha \in [0,1000)$.
    The discount parameter is fixed at $d=0$ (left column) and $d=0.5$ (right column).
    In the top row, the true probability vector $\bm{p}$ is drawn from $\mathrm{Dir}_{500}(0.1)$, while in the bottom row $\bm{p}$ is drawn from $\mathrm{Dir}_{500}(10)$.}
  \label{fig:KL}
\end{figure}

\subsubsection{Selection in the small-sample regime}
As shown in Figure~\ref{fig:KL}, for a fixed value of $d$, the cross entropy and its upper bound function exhibit similar behavior as functions of $\alpha$.
This observation motivates the use of the upper bound function $f(d,\alpha)$ as a proxy for the cross entropy.
Thus, we propose to select the hyperparameters $(d,\alpha)$ by minimizing $f(d,\alpha)$, and denote the resulting values by $(\hat d,\hat\alpha)$.

\begin{thm}\label{thm:critical}
    If the right-hand side of \eqref{ueq:upper} denoted by $f(d,\alpha)$ has critical points, then the critical points are given by
    \begin{align*}
        & d =\frac{-b\pm\sqrt{b^2-4ac}}{2a},\\
        & \alpha =\frac{T(1-d)}{C_{01}}-N,
    \end{align*}
    where $C_{01},F,a,b,c$ are defined as 
    \begin{align*}
        &C_{01}=C_0+C_1,\\
        &F=\frac{C_{0}}{2}(K-T+1),\\
        &a=(C_{01}-1)^{2}T^{2},\\
        &b=\{(2T-2C_{01}N+C_{01})(C_{01}-1)+FC_{01}\}T,\\
        &c=(T-C_{01}N+C_{01})(T-C_{01}N)-TFC_{01}.
    \end{align*}
\end{thm}

\begin{proof}
The partial derivatives of the upper bound function $f(d,\alpha)$ with respect to $\alpha$ and $d$ are given by
\begin{align}
    & \frac{\partial f}{\partial \alpha} = \frac{1}{N+\alpha}+F\left(\frac{1}{\alpha+Td+1}-\frac{1}{\alpha+Td}\right) , \label{eq:hen1}\\
    & \frac{\partial f}{\partial d} = \frac{C_{01}}{1-d}+TF\left(\frac{1}{\alpha+Td+1}-\frac{1}{\alpha+Td}\right).\label{eq:hen2}
\end{align}
If $(d,\alpha)$ is a critical point of $f$, which means that $\partial f/\partial\alpha=\partial f/\partial d=0$, then $\alpha$ and $d$ satisfy
\begin{equation}
    \alpha = \frac{T(1-d)}{C_{01}}-N\label{eq:a_and_d}.
\end{equation}
Hereafter, we refer to \eqref{eq:a_and_d} as the critical point condition. 
If $f(d,\alpha)$ has a critical point on the domain $\alpha > -d$ and $d \in [0,1)$,  then $f(d)$ has a critical point in the interval
\begin{equation*}
    0\le d<\frac{T-C_{01}N}{T-C_{01}}.
\end{equation*}
In other words, for $(d,\alpha)$ to have a critical point within this interval, it is necessary that $C_{01} < T/N$.

Under the critical point condition, $f(d,\alpha)$ can be rewritten as
\begin{equation}
    f(d)=(1-C_{01})\log(1-d)+\log\frac{T}{C_{01}}+
    F\log\frac{(C_{01}-1)Td+T-C_{01}N+C_{01}}{(C_{01}-1)Td+T-C_{01}N}\label{eq:upd}.
\end{equation}
Differentiating \eqref{eq:upd} with respect to $d$ gives
\begin{equation*}
    \frac{\mathrm{d} f}{\mathrm{d} d}=\frac{C_{01}-1}{(1-d)\{(C_{01}-1)Td+T-C_{01}N+C_{01}\}\{(C_{01}-1)Td+T-C_{01}N\}}(ad^{2}+bd+c),
\end{equation*}
where
\begin{align*}
    &a =(C_{01}-1)^{2}T^{2} ,\\
    &b =\{(2T-2C_{01}N+C_{01})(C_{01}-1)+FC_{01}\}T,\\
    &c =(T-C_{01}N+C_{01})(T-C_{01}N)-TFC_{01}.
\end{align*}
Hence, the critical points are given by
\begin{equation*}
    d_{\mathrm{critical}}=\frac{-b\pm\sqrt{b^2-4ac}}{2a}.
\end{equation*}
This completes the proof.
\end{proof}

By Theorem~\ref{thm:critical}, the function$f(d,\alpha)$ has a minimum at critical points of \eqref{eq:upd} or on the boundary.
Regarding the boundary, we consider the minimum with respect to $\alpha$ for $d=0$ or for values of $d$ near $1$. 
If critical points with respect to $\alpha$ exist, then from \eqref{eq:hen1} we have
\begin{equation}
    \alpha=\frac{-(2Td-F+1)+\sqrt{(1-F)^2+4F(N-Td)}}{2}\label{eq:a_min},
\end{equation}
showing that $f(d,\alpha)$ has a minimum in $\alpha \in \mathbb{R}$. 
If $\hat{\alpha}\le-d$, we set $\hat{\alpha} = -d+\epsilon$ and regard it as the minimum on the boundary. 
Hence, $(d,\alpha)$ is selected among these candidates so that $f(\hat d, \hat \alpha)$ is minimized. 
However, since the solutions at critical points or on the boundary involve the unknown constants $K$ and the random variables $C_1$ and $C_0$, these quantities must be estimated.

\subsubsection{Estimation of $C_1$, $C_0$ and $K$}
The quantity $1-C_0$ is known as the sample coverage and has been widely studied; see, for example, \cite{ChaoLee_1992,Esty_1983,Zhang_2009}.
In this study, we employ the so-called Good--Turing estimator introduced in \cite{Good_1953}. 
Let $m_1=\sum_{i=1}^T I(n_i=1)$ denote the number of species observed exactly once that are called singletons, where $I(\cdot)$ is the indicator function.
Using $m_1$, we estimate $C_0$ as
\begin{equation}\label{eq:GT}
    \hat{C}_0=\frac{m_1}{N}.
\end{equation}
The Good--Turing estimator \eqref{eq:GT} has been theoretically studied in, for example, \cite{Esty_1983,McAllester_2000,Zhang_2009} and is considered a reliable method for estimating unseen probabilities. 
Based on \eqref{eq:GT}, we further estimate $C_1$ as
\begin{equation*}
    \hat{C}_1=(1-\hat{C}_0)\frac{m_1}{N}.
\end{equation*}
In other words, we take the proportion of the singletons to the observed probability mass as an estimate of $C_1$. 
Similarly, we estimate $K$ using \eqref{eq:GT} by
\begin{equation*}
    \hat{K}=\frac{N}{1- \hat{C}_0}.
\end{equation*} 
By substituting these estimates into the upper bound function, we obtain the corresponding values $(\hat{d},\hat{\alpha})$, which we use as the hyperparameters $(d,\alpha)$.

\begin{rem}
In this study, we use $\hat K$ for selecting $(d,\alpha)$.
However, this does not imply that the dimensionality of the model, i.e., the number of species $K$, is being fixed. 
Our model is based on a Bayesian nonparametric framework that handles infinite-dimensional probability vectors, and $\hat{K}$ is used merely as a practical aid for analysis and computation.
\end{rem}

\subsubsection{Selection in the large-sample regime}
The advantage of using DPYM is assigning positive probabilities to unobserved species. 
However, when the sample size $N$ is sufficiently large, we typically have $T=K$, meaning that all species have been observed. 
In this case, the terms corresponding to unobserved probabilities, $C_0$ and $C_1$, vanish. 
Furthermore, in this subsection, we assume that the indices of the observed counts $(n_1,\ldots,n_K)$ are aligned with the probability vector $\bm{p}=(p_1,\ldots,p_K)$. 
The upper bound function then reduces to
\[
f(d,\alpha)=\log(N+\alpha),
\]
which depends solely on $N$ and thus provides no meaningful information for determining $(d,\alpha)$. 
Hence, when $N$ is large, we adopt a strategy of directly minimizing the cross entropy to select the hyperparameters $(d,\alpha)$ rather than relying on the upper bound function.
The cross entropy $H(\bm{p},\bm{q})$ can be rewritten as
\[
    H(\bm{p},\bm{q})
    = -\sum_{i=1}^{K} p_i \log\left(\frac{n_i - d}{N + \alpha}\right).
\]
It is natural to select $(d,\alpha)$ so as to minimize the cross entropy $H(\bm{p},\bm{q})$.
The cross entropy is minimized when
\[
    \frac{n_i - d}{N + \alpha} = p_i.
\]
We can decompose this expression as
\[
    \frac{n_i - d}{N+\alpha}
    = \frac{N}{N+\alpha} \cdot \frac{n_i}{N} \cdot \frac{n_i-d}{n_i}.
\]
By the law of large numbers, $n_i/N \to^p p_i$, $i=1,\ldots,K$, so it is preferable that
\[
    \frac{N}{N+\alpha} \approx 1, \qquad
    \frac{n_i-d}{n_i} \approx 1,
\]
where $\to^p$ denotes convergence in probability. 
From the above, it is reasonable to choose $d$ and $\alpha$ as small as possible.
In practice, for a sufficiently small $\epsilon>0$, one may choose $(d,\alpha)=(\epsilon,0)$ or $(0,\epsilon)$.
Generally, since $N > n_i$, it is natural to set
\[
    d = 0.
\]
However, to avoid the correction term being excessively small and to allow incorporation of prior information when available, we consider a more general setting $(d_0,\alpha_0)$.
In this study, based on \eqref{eq:GT}, we use the presence of singletons to decide whether $N$ is sufficiently large.
Specifically, when $m_1=0$, we set $(\hat d,\hat \alpha)=(d_0,\alpha_0)$; otherwise, when $m_1>0$, we set $(\hat d,\hat \alpha)$ by using the procedure described above.

\section{Proposed estimator}\label{sec:proposed_estimator}

Based on the preparation above, we summarize the construction of the proposed estimator $\hat H$.
Let $T$ denote the number of observed species and $n_i$ the count of the $i$-th observed species for $i = 1,\ldots,T$, with $\mathbf{y}=(n_{1},\dots,n_{T})^{\top}$.
The sample size $N$ then satisfies $N = \sum_{i=1}^Tn_i$.

First,we describe the procedure for selecting hyperparameters $(d,\alpha)$.
Let $m_1=\sum_{i=1}^T I(n_i=1)$ denote the number of singletons.
If $m_1=0$, the hyperparameters are set to $(\hat d,\hat\alpha)=(d_0,\alpha_0)$, where $(d_0,\alpha_0)$ are predetermined values based on prior information or other considerations independent of the sample.
Next, we describe the selection of $(d,\alpha)$ when $m_1\neq0$.
For this purpose, we employ the following estimators for $C_0,C_1$ and $K$ given by
\[
    \hat C_0=\frac{m_1}{N},\quad
    \hat{C}_1=(1-\hat{C}_0)\frac{m_1}{N}, \quad
    \hat{K}=\frac{N}{1- \hat{C}_0}.
\]
We consider the estimated upper bound function
\begin{equation}
    \hat{f}(d,\alpha)=\log(N+\alpha)-(\hat C_{1}+\hat C_{0})\log(1-d)+\frac{\hat C_{0}}{2}(\hat K-T+1) \log\left( \frac{\alpha+Td+1}{\alpha+Td} \right)\label{eq:hat_f},
\end{equation}
and select $(d,\alpha)$ so as to minimaze $\hat{f}(d,\alpha)$.
By Theorem~\ref{thm:critical}, this gives two candidate pairs
\[
d_{\pm} = \frac{-\hat b\pm\sqrt{\hat b^2-4\hat a\hat c}}{2\hat a}, \quad
\alpha_{\pm} = \frac{T(1-d_{\pm})}{\hat C_{01}}-N,
\quad 
\]
where the signs are taken consistently.
Here,
\begin{align*}
        &\hat C_{01}=\hat C_0+\hat C_1,\quad
        \hat F=\frac{\hat C_{0}}{2}(\hat K-T+1),\\
        &\hat a=(\hat C_{01}-1)^{2}T^{2},\quad
        \hat b=\{(2T-2\hat C_{01}N+\hat C_{01})(\hat C_{01}-1)+\hat  F\hat C_{01}\}T,\\
        &\hat c=(T-\hat C_{01}N+\hat C_{01})(T-\hat C_{01}N)-T\hat F\hat C_{01}.
\end{align*}
In addition, we consider the boundary case by minimizing~\eqref{eq:hat_f} with respect to $\alpha$ when $d=0$ or when $d$ is sufficiently close to $1$.
When $d=0$, the corresponding value of $\alpha$ is
\begin{equation*}
    \alpha_0=\frac{\hat F-1+\sqrt{(1-\hat F)^2+4\hat FN}}{2}.
\end{equation*}
When $d$ is sufficiently close to $1$, i.e., $d=1-\epsilon$ for small $\epsilon>0$, the corresponding value is 
\begin{equation*}
    \alpha_1 = \frac{-(2T(1-\epsilon)-\hat F+1)+\sqrt{(1-\hat F)^2 + 4 \hat{F} \{N-T(1-\epsilon)\} }}{2}.
\end{equation*}
We define a candidate set of hyperparameters
\begin{equation*}
    \Theta':=\{(d_+,\alpha_+),(d_-,\alpha_-),(d_0,0),(1-\epsilon,\alpha_1)\},
\end{equation*}
and select the hyperparameters $(d,\alpha)$ as
\begin{equation*}
    (\hat d,\hat\alpha)=\argmin_{(d,\alpha)\in\Theta'} \hat f(d,\alpha).
\end{equation*}

Next, using the selected hyperparameters $(\hat d,\hat\alpha)$, we construct an estimator of the Shannon entropy $H(\bm{p})$ of the population probability vector $\bm{p} \in \Delta_K$.
Substituting $(\hat d,\hat\alpha)$ into \eqref{eq:qstar}, we set
\[
\hat{\bm{q}}^* =\left( \frac{n_1 - \hat{d}}{N + \hat{\alpha}} ,\dots, \frac{n_T - \hat{d}}{N  + \hat{\alpha}}, \frac{\alpha + T\hat{d}}{N + \hat{\alpha}} \right).
\]
The Shannon entropy of $\hat{\bm{q}}^* \in \Delta_{T+1}$ can be expressed as
\begin{equation*}
    H(\hat{\bm{q}}^*)=-\sum_{t=1}^T\frac{n_t-\hat d}{N+\hat\alpha}\log\left(\frac{n_t-\hat d}{N+\hat\alpha}\right)-\frac{\hat\alpha+T\hat d}{N+\hat\alpha}\log\left(\frac{\hat\alpha+T\hat d}{N+\hat\alpha}\right).
\end{equation*}
For the Shannon entropy of marginal PYP, when $\hat d=0$, we have
\begin{equation*}
    H(\hat{\bm{\pi}})=(1+\hat\alpha)\log(1+\hat\alpha)-\hat\alpha\log\hat\alpha \ ;
\end{equation*}
and when $\hat d\neq0$, we use the approximation
\begin{align*}
    H(\hat{\bm{\pi}})
    &=-\log\left(\frac{(1-\hat{d})\Gamma(\frac{\hat{\alpha}+1}{\hat{d}})}{\hat{\alpha}\Gamma(\frac{\hat{\alpha}}{\hat{d}})}\right)-\frac{(1-\hat{d})\Gamma(\frac{\hat{\alpha}+1}{\hat{d}})}{\hat{\alpha}\Gamma(\frac{\hat{\alpha}}{\hat{d}})}\sum_{k=1}^{n}\left(\frac{\Gamma(\frac{\hat{\alpha}}{\hat{d}}+k)}{\Gamma(\frac{\hat{\alpha}+1}{\hat{d}}+k)}\right)\log\left(\frac{\Gamma(\frac{\hat{\alpha}}{\hat{d}}+k)}{\Gamma(\frac{\hat{\alpha}+1}{\hat{d}}+k)}\right)\\
    &\qquad - \frac{(1-\hat{d})\Gamma(\frac{\hat{\alpha}+1}{\hat{d}})}{\hat{\alpha}\Gamma(\frac{\hat{\alpha}}{\hat{d}})} \left\{ \frac{(n+1)^{\frac{\hat{d}-1}{\hat{d}}}}{\hat{d}-1}\log(n+1)-\frac{\hat{d}(n+1)^{\frac{\hat{d}-1}{\hat{d}}}}{(\hat{d}-1)^2} \right\}
\end{align*}
for sufficiently large $n$.
Using $H(\hat{\bm{q}}^*)$ and $H(\hat{\bm{\pi}})$ described above, we propose the estiamtor $\hat{H}$ of $H(\bm{p})$ given by
\begin{equation*}
    \hat H=H(\hat{\bm{q}}^*)+\frac{\hat\alpha+T\hat d}{N+\hat\alpha}H(\hat{\bm{\pi}}).
\end{equation*}

\begin{rem}
In the proposed method, we estimate the entropy $H(\bm{p})$ of a population distribution $\bm{p} \in \Delta_K$ using a model distribution $\bm{q} \in \Delta_{\infty}$ that does not coincide with the true distribution.
The proposed estimator selects hyperparameters to mitigate the effect of the inherent model mismatch, in order to maintain good estimation accuracy in both small-sample, high-dimensional regimes and large-sample scenarios.
As we show later, the estimator is consistent as $N\to\infty$.
\end{rem}

\begin{rem}
Archer et al.~\cite{Archer2014} derived the posterior mean of the Shannon entropy $\mathsf{E}[H \mid \mathbf{y}, d, \alpha]$ based on the DPYM.
They further assumed a prior distribution with density $\tilde{p}$ over $(d,\alpha)$ and proposed the estimator
\[
\hat{H}^{\mathrm{PYM}}
= \int \mathsf{E}[H \mid \mathbf{y}, d, \alpha]\,
\frac{p(\mathbf{y}\mid d,\alpha) \tilde{p}(d,\alpha)}{p(\mathbf{y})}
\mathrm{d}(d,\alpha),
\]
where
\begin{align*}
& p(\mathbf{y}\mid d,\alpha)
= \frac{\left\{ \prod_{l=1}^{T-1}(\alpha+ld) \right\}\left\{ \prod_{i=1}^{T}\Gamma(n_i-d)\right\} \Gamma(1+\alpha)}{\Gamma(1-d)^T\Gamma(\alpha+N)}, \\
& p(\mathbf{y}) = \int p(\mathbf{y}\mid d,\alpha) \tilde{p}(d,\alpha) \mathrm{d}(d,\alpha).
\end{align*}
The estimator $\hat{H}^{\mathrm{PYM}}$is known to behave unstably when all observed species are singletons \cite{Archer2014}, and it is computationally expensive due to the required numerical integration.
Moreover, the choice of prior with respect to $(d,\alpha)$ can strongly affect the results, and it is generally difficult to specify a suitable prior in advance.
\end{rem}

\section{Finiteness of entropy for regularly varying distributions}\label{sec:finiteness}
\subsection{Regularly varying function and regularly varying distribution}

The proposed method is designed for settings in which the sample size $N$ is relatively small compared to the true number of species $K$. 
To ensure that the method retains a certain level of accuracy when $N$ is sufficiently large, we establish its consistency in the large-sample regime. 
The concepts of regularly varying functions and regularly varying distributions are used to prove this result, which are needed to verify the finiteness of the entropy.
We summarize the relevant fundamentals in this subsection.
For further details on regularly varying distributions, see Nair et al.~\cite{Nair_2022}, Resnick~\cite{Resnick_2007}, and Bingham et al.~\cite{Bingham_1987}.

\begin{dfn}[\cite{Nair_2022}]
A function $f:\mathbb{R}_+\rightarrow\mathbb{R}_+$ is regularly varying of index $\rho\in\mathbb{R}$, denoted $f\in\mathcal{RV}(\rho)$, if for any $y>0$,
\begin{equation*}
    \lim_{x\rightarrow\infty}\frac{f(xy)}{f(x)}=y^\rho.
\end{equation*}
In particular, when $\rho = 0$, the function $f$ is called slowly varying.
\end{dfn}
A function $f$ is regularly varying with index $\rho$ if and only if it can be represented as $f(x)=x^\rho L(x)$ $(x \in \mathbb{R}_+)$ where $L$ is a slowly varying function~\cite{Nair_2022}.

\begin{dfn}
A random variable $Z$ is said to follow a regularly varying distribution with index $\rho$ if its complementary cumulative distribution function or tail probability $\bar F(z)=P(Z\ge z)$ $(z \geq 0)$ is a regularly varying function with index $\rho$.
\end{dfn}

We briefly review the uniform convergence theorem and the monotone density theorem, which are used in the proof of Proposition~\ref{prop:PMF} below.

\begin{thm}[Uniform convergence theorem~\cite{Bingham_1987}]\label{thm:UCT}
    Let $L$ be a slowly varying function. Then, for any compact set $\Lambda \subset (0, \infty)$,
    \begin{equation*}
        \lim_{x\to\infty}\sup_{\lambda\in\Lambda}\left|\frac{L(\lambda x)}{L(x)}-1\right|=0.
    \end{equation*}
\end{thm}

\begin{thm}[Monotone density theorem~\cite{Nair_2022}]\label{thm:monotone}
Suppose that the function $f$ is absolutely continuous with derivative $f'$. 
If $f\in\mathcal{RV}(\rho)$ and $f'$ is eventually monotone, then $xf'(x)\sim\rho f(x)$.
Moreover, if $\rho\neq0$, then $|f'|\in\mathcal{RV}(\rho-1)$.
\end{thm}
Here, a function $g$ is said to be eventually monotone if there exists $x_0>0$ such that $g$ is monotone over $[x_0,\infty]$.
Similarly, if the corresponding condition holds for a sequence, we say that the sequence is eventually monotone.

In the following proposition, we show that, under suitable assumptions, the probability mass function of a discrete regularly varying distribution is a regularly varying function.
\begin{prop}\label{prop:PMF}
Let $\mathrm{D}$ be a discrete regularly varying distribution on $\mathbb{N}$ with index $-\rho$. 
Let $\bm{p} = (p_1, p_2, \ldots) \in \Delta_{\infty}$ be defined by $p_k = \mathrm{D}(k)$ for $k = 1,2,\ldots$.
Assume that there exists $n_0 \in \mathbb{N}$ such that $p_{n_0}\ge p_{n_0+1}\ge\ldots$.
Then the function $f$ defined by 
\begin{equation}\label{eq:funk_from_pmf}
f(t) = \sum_{i=1}^\infty p_iI\{i-1< t\le i\}  \quad (t \in \mathbb{R}_+)
\end{equation}
is regularly varying with index $-(\rho+1)$.
\end{prop}

\begin{proof} 
    The tail distribution of $\mathrm{D}$, defined by $\bar F(z)=\sum_{i\ge z} p_i$ $(z \geq 0)$, is a step function and by assumption, satisfies $\bar F \in \mathcal{RV}(-\rho)$.
    Define the function $\tilde F$ by
    \begin{align*}
        \tilde F(z)
        &=p_{\lfloor z\rfloor}(\lceil z\rceil-z)+\sum_{i=\lceil z\rceil}^\infty p_i\\
        &=p_{\lfloor z\rfloor}(\lceil z\rceil-z)+\bar F(\lceil z\rceil) \quad (z \geq 0).
    \end{align*}
    Here, $\lfloor\cdot\rfloor$ and $\lceil \cdot\rceil$ denote the floor and ceiling functions, respectively.
    
    We first show that $\tilde F \in \mathcal{RV}(-\rho)$.
    For any $z \ge 0$, we have
    \begin{align*}
        \frac{\tilde F(z)}{\bar F(z)}
        &=\frac{p_{\lfloor z\rfloor}(\lceil z\rceil-z)+\bar F(\lceil z\rceil)}{\bar F(z)}\\
        &\le\frac{p_{\lfloor z\rfloor}+\bar F(z)}{\bar F(z)}\\
        &=\frac{\bar F(\lfloor z\rfloor)-\bar F(\lfloor z\rfloor+1)+\bar F(\lceil z\rceil)}{\bar F(z)}\\
        &\le\frac{\bar F(z-1)}{\bar F(z)}.
    \end{align*}
    From Theorem~\ref{thm:UCT}, it follows that
    \[
    \lim_{z\to\infty}\frac{\bar F((1-1/z)z)}{\bar F(z)}=1.
    \]
    Therefore, it holds that
    \[
    \lim_{z\to\infty}\frac{\tilde F(z)}{\bar F(z)}=1.
    \]
    For any $\lambda>0$, we thus have
    \begin{align*}
        \lim_{z\to\infty}\frac{\tilde F(\lambda z)}{\tilde F(z)}
        &=\lim_{z\to\infty}\frac{\tilde F(\lambda z)}{\bar F(\lambda z)}
        \frac{\bar F(\lambda z)}{\bar F(z)}
        \frac{\bar F(z)}{\tilde F(z)}\\
        &=\lambda^{-\rho} .
    \end{align*}
    Hence, $\tilde F\in\mathcal{RV}(-\rho)$.

    Next, we show that the function $f$ defined in \eqref{eq:funk_from_pmf} belongs to $\mathcal{RV}(-(\rho+1))$.
    Note that
    \[
    \tilde F(z)=\int_z^\infty f(t)dt.
    \]
    Since $\tilde F(0) = 1$, the function $\tilde F$ is absolutely continuous and satisfies $\tilde F'(x) = -f(x)$ almost everywhere.
    By the assumption that the sequence $\{p_n\}_{n \ge 1}$ is eventually monotone, the function $-f$ is also eventually monotone.
    Therefore, by Theorem~\ref{thm:monotone}, it follows that $f \in \mathcal{RV}(-\rho-1)$.
\end{proof}
The following proposition about regularly varying distributions is known \cite{Nair_2022}.
\begin{prop}\label{prop:index_of_RV}
    If a random variable $Z$ follows a regularly varying distribution with index $-\rho$, then $\rho \ge 0$.
\end{prop}

By Propositions~\ref{prop:PMF} and~\ref{prop:index_of_RV}, for a discrete distribution whose tail probability is regularly varying, if the assumption of Proposition~\ref{prop:PMF} is satisfied, then the function $f:\mathbb{R}_+ \to \mathbb{R}_+$ defined by \eqref{eq:funk_from_pmf} satisfies
\[
    f \in \mathcal{RV}(-\rho)
\]
for some $\rho \ge 1$.
Hereafter, throughout this paper, we assume without further mention that
any discrete regularly varying distribution has an eventually monotone
probability mass function and that the associated function $f$ belongs to
$\mathcal{RV}(-\rho)$ with $\rho \ge 1$.

The following result, called Karamata's representation theorem, is well known for regularly varying functions~\cite{Nair_2022}.

\begin{thm}[Karamata's representation theorem~\cite{Nair_2022}]\label{thm:Karamata}
A function $g:\mathbb{R}_+\rightarrow\mathbb{R}_+$ is a regularly varying function with index $\rho$ if and only if it can be represented as 
\begin{equation*}
g(x)=c(x)\exp\left\{\int_1^x \frac{\beta(t)}{t}\mathrm{d}t\right\}
\end{equation*}
for $x>0$, where $c(\cdot)$ and $\beta(\cdot)$ satisfy $\lim_{x\rightarrow\infty}c(x)=c\in(0,\infty)$ and $\lim_{x\rightarrow\infty}\beta(x)=\rho$.
\end{thm}

\subsection{Finiteness of the Shannon entropy of regularly varying distributions}

As shown later, the marginal PYP induces a regularly varying distribution. 
However, there exist regularly varying distributions whose Shannon entropy diverges \cite{Baccetti_2013}.
In this subsection, we first establish general conditions for the finiteness of the Shannon entropy of regularly varying distributions, and then apply these results to show that the Shannon entropy of the marginal PYP is finite.

\begin{thm}\label{thm:powerlaw_ent}
Let $\rho \ge 1$ and $\bm{p} = (p_1, p_2, \ldots) \in \Delta_{\infty}$.
Suppose that $\bm{p}$ is eventually monotone and that the function $f$ defined in \eqref{eq:funk_from_pmf} is regularly varying with index $-\rho$.
Then the convergence or divergence of the Shannon entropy $H(\bm{p})$
is characterized as follows:
\begin{align*}
    \text{if}\quad\rho > 1, & \quad \text{then} \quad H(\bm{p}) < \infty, \\
    \text{if}\quad\rho = 1 \quad \text{and} \quad \sum_{k=1}^\infty p_k \log k < \infty,
    &\quad \text{then} \quad H(\bm{p}) < \infty, \\
    \text{if}\quad\rho = 1 \quad \text{and} \quad \sum_{k=1}^\infty p_k \log k = \infty,
    &\quad \text{then} \quad H(\bm{p}) = \infty.
\end{align*}
\end{thm}

\begin{proof}
By the properties of regularly varying functions, there exists a slowly varying function $L(\cdot)$ such that
\begin{equation*}
    f(t)=t^{-\rho}L(t) \quad (t\in\mathbb{R_+}).
\end{equation*}
In particular, for $k\in\mathbb{N}$,
\[
f(k)=p_k=k^{-\rho}L(k).
\]
Therefore, the conclusion follows from Lemmas~\ref{lem:finiteness1} and~\ref{lem:finiteness2} below.
\end{proof}

\begin{lem}\label{lem:finiteness1}
Let $\rho \geq 1$.
Consider $\bm{p} = (p_1,p_2,\ldots) \in \Delta_\infty$ such that $p_k = k^{-\rho}L(k)$ $(k=1,2,\ldots)$, where $L$ is a slowly varying function.
Then, for the Shannon entropy $H(\bm{p})$ of $\bm{p}$,
\begin{align*}
    \text{if}\quad\sum_{k=1}^\infty p_k\log k<\infty,
    &\quad \text{then} \quad H(\bm{p})<\infty, \\
    \text{if}\quad \sum_{k=1}^\infty p_k\log k=\infty,
    &\quad \text{then} \quad H(\bm{p})=\infty.
\end{align*}
\end{lem}

\begin{proof}
    By Theorem~\ref{thm:Karamata}, for any $\epsilon>0$, there exists $n_0\in\mathbb{N}$ such that  $|\beta(k)|<\epsilon$ and $|c(k)-c|<\epsilon$ for all $k\ge n_0$.
    For $k \geq n_0$, this implies
    \begin{equation*}
        \log(c-\epsilon)+\int_1^{n_0}\frac{\beta(t)}{t}dt-\epsilon\log\left(\frac{k}{n_0}\right)\le \log L(k)\le\log(c+\epsilon)+\int_1^{n_0}\frac{\beta(t)}{t}dt+\epsilon\log\left(\frac{k}{n_0}\right).
    \end{equation*}
    On the other hand, the Shannon entropy $H(\bm{p})$ can be decomposed as
    \begin{equation*}
       H(\bm{p})=-\sum_{k=1}^{n_0}p_k\log p_k+\sum_{k=n_0+1}^{\infty}\frac{L(k)}{k^\rho}\log\left( \frac{k^\rho}{L(k)}\right).
    \end{equation*}
    The first term on the right-hand side is finite.
    For the second term, we have
    \begin{align*}
        \sum_{k=n_0+1}^{\infty}p_k\log\left( \frac{k^\rho}{L(k)}\right)
        &=\sum_{k=n_0+1}^{\infty}p_k(\rho\log k-\log L(k))\\
        &\le\sum_{k=n_0+1}^{\infty}p_k\left\{\rho\log k-\log(c-\epsilon)-\int_1^{n_0}\frac{\beta(t)}{t}dt+\epsilon\log\left(\frac{k}{n_0}\right)\right\}\\
        &=(\rho+\epsilon)\sum_{k=n_0+1}^{\infty}p_k\log k+A\sum_{k=n_0+1}^{\infty}p_k,
    \end{align*}
    where $A$ is a constant.
    Hence, finiteness of the first term on the right-hand side implies finiteness of $H(\bm{p})$.
    The divergence case follows from an analogous argument.
    This completes the proof.
\end{proof}

\begin{lem}\label{lem:finiteness2}
    Let $\rho > 1$.
    Consider a sequence $\{p_k\}_{k\geq 1}$ such that $p_k = k^{-\rho}L(k)$ $(k=1,2,\ldots)$, where $L$ is a slowly varying function.
    Then, $\sum_{k=1}^\infty p_k\log k<\infty$.
\end{lem}

\begin{proof}
    Fix $\epsilon\in(0, \rho-1)$.
    By Theorem~\ref{thm:Karamata}, there exist functions $\beta(\cdot)$ and $c(\cdot)$ and an integer $n_0\in\mathbb{N}$ such that, for all $k\ge n_0$, both $|\beta(k)+\rho|<\epsilon$ and $|c(k)-c|<\epsilon$ hold.
    For $k \geq n_0$, we have
    \begin{align*}
        p_k\log k
        &=c(k)\exp\left\{\int_1^k\frac{\beta(t)}{t}dt\right\}\log k\\
        &\le(c+\epsilon)\exp\left\{\int_1^{n_{0}}\frac{\beta(t)}{t}dt+\int_{n_0}^k\frac{-\rho+\epsilon}{t}dt\right\}\log k\\
        &=(c+\epsilon)\exp\left\{\int_1^{n_{0}}\frac{\beta(t)}{t}dt+(-\rho+\epsilon)\log\frac{k}{n_0}\right\}\log k\\
        &=A\cdot\frac{\log k}{k^{\rho-\epsilon}},
    \end{align*}
    where $A$ is a constant.
    Therefore, it holds that
    \begin{equation*}
        \sum_{k=1}^\infty p_k\log k<\sum_{k=1}^\infty\frac{\log k}{k^{\rho-\epsilon}}<\infty.
    \end{equation*}
    This completes the proof.
\end{proof}

\begin{rem}
Let $u > 0$.
Baccetti and Visser~\cite{Baccetti_2013} consider a discrete distribution with the infinite-dimensional probability vector $(p_{3},p_{4},\ldots)$ given by
\begin{equation*}
    p_k=\frac{1}{\Sigma(u)k(\log k)^{1+u}} \quad (k=3,4,\ldots),
\end{equation*}
where $\Sigma(u)$ is a normalizing constant.
The associated slowly varying function is $L(k)=1/ \{\Sigma(u)(\log k)^{1+u} \}$.
In this case,
\begin{align*}
    \sum_{k=3}^\infty\frac{L(k)\log k}{k}
    &=\sum_{k=3}^\infty\frac{1}{\Sigma(u)k(\log k)^{u}}. \quad 
\end{align*}
The series diverges for $u\in(0,1]$ and converges for $u>1$.
A corresponding convergence criterion is also established in Baccetti and Visser~\cite{Baccetti_2013}.
\end{rem}

Next, we show that the marginal PYP follows a regularly varying distribution.
As a consequence, we establish that the Shannon entropy of the marginal PYP is finite.

\begin{prop}\label{prop:regularly_PYP}
Assume that $d\neq0$.
Then the probability mass function of the marginal PYP, denoted by $P_{d,\alpha}^{\mathrm{MPY}}(k)$ $(k=1,2,\ldots)$, is a regularly varying function with index $-1/d$.
\end{prop}

\begin{proof}
Since $d\neq0$, let $p_k = P_{d,\alpha}^\mathrm{MPY}(k)$ be given by \eqref{eq:MPY2} for $k=1,2,\ldots$, and let $f$ be defined by \eqref{eq:funk_from_pmf} using $\bm{p}=(p_1,p_2,\ldots)$.
For all $\lambda>0$ and $t\in\mathbb{R_+}$, it holds that
\[
\frac{f(\lambda t)}{f(t)}
        =\frac{\Gamma(\frac{\alpha}{d}+\lceil\lambda t\rceil)\Gamma(\frac{\alpha+1}{d}+\lceil t\rceil)}{\Gamma(\frac{\alpha}{d}+\lceil t\rceil)\Gamma(\frac{\alpha+1}{d}+\lceil\lambda t\rceil)}.
\]
We have
\begin{align*}
    \frac{f(\lambda t)}{f(t)}
    &=\frac{(\frac{\alpha}{d}+\lceil\lambda t\rceil)^{\frac{1}{d}}\Gamma(\frac{\alpha}{d}+\lceil\lambda t\rceil)}{\Gamma(\frac{\alpha+1}{d}+\lceil\lambda t\rceil)}
        \frac{\Gamma(\frac{\alpha+1}{d}+\lceil t\rceil)}{(\frac{\alpha}{d}+\lceil t\rceil)^{\frac{1}{d}}\Gamma(\frac{\alpha}{d}+\lceil t\rceil)}
        \frac{(\frac{\alpha}{d}+\lceil t\rceil)^{\frac{1}{d}}}{(\frac{\alpha}{d}+\lceil\lambda t\rceil)^{\frac{1}{d}}}\\
    &\le\frac{(\frac{\alpha}{d}+\lceil\lambda t\rceil)^{\frac{1}{d}}\Gamma(\frac{\alpha}{d}+\lceil\lambda t\rceil)}{\Gamma(\frac{\alpha+1}{d}+\lceil\lambda t\rceil)}
        \frac{\Gamma(\frac{\alpha+1}{d}+\lceil t\rceil)}{(\frac{\alpha}{d}+\lceil t\rceil)^{\frac{1}{d}}\Gamma(\frac{\alpha}{d}+\lceil t\rceil)}\frac{(\frac{\alpha}{d}+ t+1)^{\frac{1}{d}}}{(\frac{\alpha}{d}+\lambda t)^{\frac{1}{d}}}.
\end{align*}
By an analogous argument, we have
\[
\frac{f(\lambda t)}{f(t)}
\ge\frac{(\frac{\alpha}{d}+\lceil\lambda t\rceil)^{\frac{1}{d}}\Gamma(\frac{\alpha}{d}+\lceil\lambda t\rceil)}{\Gamma(\frac{\alpha+1}{d}+\lceil\lambda t\rceil)}
        \frac{\Gamma(\frac{\alpha+1}{d}+\lceil t\rceil)}{(\frac{\alpha}{d}+\lceil t\rceil)^{\frac{1}{d}}\Gamma(\frac{\alpha}{d}+\lceil t\rceil)}\frac{(\frac{\alpha}{d}+ t)^{\frac{1}{d}}}{(\frac{\alpha}{d}+\lambda t+1)^{\frac{1}{d}}}.
\]
Applying Stirling's formula and taking the limit $t\to\infty$ yield
\[
\lim_{t\to\infty}\frac{f(\lambda t)}{f(t)}=\lambda^{-\frac{1}{d}}.
\]
Therefore, $f\in\mathcal{RV}(-1/d)$.
This completes the proof.
\end{proof}

\begin{cor}\label{cor:bdd_PYP}
The Shannon entropy of the marginal PYP is finite.
\end{cor}

\begin{proof}
If $d=0$, the marginal PYP reduces to a geometric distribution, and its Shannon entropy is finite.  
If $d\neq0$, the marginal PYP has a regular variation index of $-1/d$, and Theorem~\ref{thm:powerlaw_ent} together with Proposition~\ref{prop:regularly_PYP} implies that the Shannon entropy of the marginal PYP is finite because $d\in(0,1)$.
\end{proof}

\section{Consistency of the proposed entropy estimator}\label{sec:consistency}

In this section, we establish the consistency of the proposed estimator.
For clarity, we reintroduce the notation used in the asymptotic analysis.  
Let $K \in \mathbb{N}$ and consider a probability vector $\bm{p} = (p_1, \ldots, p_K) \in \Delta_K$.  
Let $d_0 \in [0,1)$ and $\alpha_0 \in (-d_0, \infty)$ be fixed constants.  
Consider a random sample of size $N \in \mathbb{N}$ drawn from the population distribution with frequencies given by $\bm{p}$.  
Let $\mathbf{x}_N = (\tilde{n}_{N,1}, \ldots, \tilde{n}_{N,K})$ denote the observed labeled frequencies, $T_N$ the number of observed species, $\mathbf{y}_N = (n_{N,1}, \dots, n_{N,T_N})$ the observed unlabeled frequencies, and $m_{N,1}$ the number of singletons.  
In the proposed method, the parameters $(d, \alpha)$ estimated from the observations are denoted by $(\hat{d}_N, \hat{\alpha}_N)$.
For $d \in [0,1)$ and $\alpha \in (-d, \infty)$, define
\[\pi_{N,d,\alpha,k} = \frac{1-d}{\alpha + T_N d}\prod_{j=1}^{k}\frac{\alpha+(T_N + j - 1)d}{\alpha+(T_N + j - 1)d+1} \quad (k=1,2,\ldots),\]
and let $\boldsymbol{\pi}_{N,d,\alpha} = (\pi_{N,d,\alpha,1},\pi_{N,d,\alpha,2},\ldots)$.
Further define
\begin{align*}
&\bm{q}^*_{N,d,\alpha}(\mathbf{y}_N) =\left( \frac{n_{N,1} - d}{N+\alpha} ,\dots, \frac{n_{N,T_N} - d}{N+\alpha} , \frac{\alpha + T_N d}{N+\alpha} \right), \\
&\bm{q}_{N,d,\alpha}(\mathbf{y}_N) =\left( \frac{n_{N,1} - d}{N+\alpha} ,\dots, \frac{n_{N,T_N} - d}{N+\alpha} , \frac{\alpha + T_N d}{N+\alpha} \boldsymbol{\pi}_{N,d,\alpha} \right).
\end{align*}
Then, the proposed estimator $\hat{H}_N$ can be written as $H(\bm{q}_{N,\hat{d}_N,\hat{\alpha}_N}(\mathbf{y}_N))$.
Consider the sequence of random variables, $\{T_N\}_{N\geq1}$, $\{\mathbf{x}_N\}_{N\geq1}$, $\{\mathbf{y}_N\}_{N\geq1}$, $\{m_{N,1}\}_{N\geq1}$, $\{(\hat{d}_N, \hat{\alpha}_N)\}_{N\geq1}$, and $\{\hat{H}_N\}_{N\geq1}$.
Hereafter, we take the limit $N\to\infty$, and establish the consistency of the proposed estimator.

\begin{prop}
The estimator $\hat{H}_N$ is consistent for $H(\bm{p})$.
\end{prop}

\begin{proof}

By the law of large numbers, we have $N^{-1} \mathbf{x}_N \to^p \boldsymbol{p}$.
Since $p_k>0$ for all $k\in\{1,\ldots,K\}$, it follows that $\mathsf{P}(m_{N,1} = 0) \to 1$, and consequently $\mathsf{P}((\hat{d}_N, \hat{\alpha}_N) = (d_0, \alpha_0)) \to 1$.
Moreover, $\mathsf{P}(T_N = K) \to 1$.
Define the events
\[
E_N = \{ (\hat{d}_N, \hat{\alpha}_N) \neq (d_0, \alpha_0) \}, \quad
\tilde{E}_N = \{ T_N \neq K \}.
\]
Then the sequences of events $\{E_N \}_{N\geq1}$ and $\{ \tilde{E}_N \}_{N\geq1}$ satisfy $\mathsf{P}(E_N) \to 0$ and $\mathsf{P}(\tilde{E}_N) \to 0$.

For any $\ve>0$, we have
\[
\mathsf{P}(| H(\bm{q}_{N,\hat{d}_N,\hat{\alpha}_N}(\mathbf{y}_N)) - H(\bm{p}) | > \ve) \\
\leq \mathsf{P}(| H(\bm{q}_{N,d_0,\alpha_0} (\mathbf{y}_N)) - H(\bm{p}) | > \ve ) + \mathsf{P}(E_N) 
\]
Since $\mathsf{P} (E_N) \to 0$, it is sufficient to show that $H(\bm{q}_{N,d_0,\alpha_0}(\mathbf{y}_N)) \to^p H(\bm{p})$.
For the Shannon entropy of $\bm{q}_{N,d_0,\alpha_0}(\mathbf{y}_N)$, we have
\begin{equation*}
    H(\bm{q}_{N,d_0,\alpha_0}(\mathbf{y}_N)) 
    = H(\bm{q}^*_{N,d_0,\alpha_0}(\mathbf{y}_N)) + \frac{\alpha_0 + T_N d_0}{N + 
    \alpha_0}H(\boldsymbol {\pi}_{N,d_0,\alpha_0}).
\end{equation*}
By Corollary~\ref{cor:bdd_PYP}, $H(\boldsymbol{\pi}_{N,d_0,\alpha_0}) < \infty$.
Moreover, it follows from $T_N \leq K < \infty$ that $H(\bm{q}_{N,d_0,\alpha_0}(\mathbf{y}_N)) - H(\bm{q}^*_{N,d_0,\alpha_0}(\mathbf{y}_N)) \to^p 0$.
Hence, it suffices to prove that $H(\bm{q}^*_{N,d_0,\alpha_0}(\mathbf{y}_N)) \to^p H(\bm{p})$.
For any $\tilde{\ve}>0$, we have
\begin{align*}
\mathsf{P}(| H(\bm{q}^*_{N,d_0,\alpha_0}(\mathbf{y}_N)) - H(\bm{p}) | > \tilde{\ve}) 
&\leq \mathsf{P}(\{| H(\bm{q}_{N,d_0,\alpha_0} (\mathbf{y}_N)) - H(\bm{p}) | > \tilde{\ve}\} \cap \{T_N = K\} ) + \mathsf{P}(\tilde{E}_N)  \\
&= \mathsf{P}(\{| H(\bm{q}_{N,d_0,\alpha_0} (\mathbf{x}_N)) - H(\bm{p}) | > \tilde{\ve}\} \cap \{T_N = K\} ) + \mathsf{P}(\tilde{E}_N),
\end{align*}
where, on the event $T_N = K$, we define
\[ \bm{q}^*_{N,d,\alpha}(\mathbf{x}_N) =\left( \frac{\tilde{n}_1 - d}{N+\alpha} ,\dots, \frac{\tilde{n}_{K} - d}{N+\alpha} , \frac{\alpha + K d}{N+\alpha} \right),\]
and use $ H(\bm{q}^*_{N,d_0,\alpha_0} (\mathbf{y}_N)) = H(\bm{q}^*_{N,d_0,\alpha_0} (\mathbf{x}_N))$.
Since $\bm{q}^*_{N,d_0,\alpha_0}(\mathbf{x}_N) \to^p (\bm{p},0)$, the continuous mapping theorem implies $H(\bm{q}^*_{N,d_0,\alpha_0}(\mathbf{x}_N)) \to^p H(\bm{p})$.
Moreover, $\mathsf{P}(\tilde{E}_N) \to 0$, which implies $H(\bm{q}^*_{N,d_0,\alpha_0}(\mathbf{y}_N)) \to^p H(\bm{p}) $.
Hence, we conclude that $\hat{H}_N = H(\bm{q}_{N,\hat{d}_N,\hat{\alpha}_N}(\mathbf{y}_N)) \to^p H(\bm{p})$.
\end{proof}

\section{Simulation studies}\label{sec:simulation}
\subsection{Conventional methods for the Shannon entropy estimation}

In this section, we investigate the behavior of the proposed estimator on synthetic data and compare it with existing methods. 
We first briefly review the baseline estimators considered for comparison.

Let $\bm{p} \in \Delta_K$ denote the population probability vector, and let $H(\boldsymbol{p})$ be the Shannon entropy.
The maximum likelihood estimator (MLE) of $H(\boldsymbol{p})$, denoted by  $\hat{H}^{\mathrm{ML}}$, is obtained by plugging the MLE $\hat{\boldsymbol{p}}^\mathrm{ML}$ of $\boldsymbol{p}$ into $H(\boldsymbol{p})$, that is,
\begin{equation*}
    \hat{H}^{\mathrm{ML}}=-\sum_{i=1}^T\frac{n_i}{N}\log\frac{n_i}{N}.
\end{equation*}
A bias-corrected version of $\hat{H}^{\mathrm{ML}}$ is given by
\begin{equation}
	\hat{H}^{\mathrm{MM}}=\hat{H}^{\mathrm{ML}}+\frac{T-1}{2N},
\end{equation}
which is known as the Miller--Madow estimator~\cite{Miller_1955}. 
Moreover, as the estimator of $\boldsymbol{p}$ the Good--Turing estimator 
\[
	\hat{\bm{p}}^{\mathrm{GT}} =\left( 1-\frac{m_{1}}{N} \right) \hat{\bm{p}}^{\mathrm{ML}}
\]
has been proposed~\cite{Good_1953}.
Substituting $\hat{\bm{p}}^{\mathrm{GT}}$ into the Horvitz--Thompson estimator yields the so-called Chao--Shen estimator
\begin{equation}
	\hat{H}^{\mathrm{CS}}=-\sum_{i=1}^{T}\frac{\hat{p}^{\mathrm{GT}}_{i}\log{\hat{p}^{\mathrm{GT}}}_{i}}{1-(1-\hat{p}^{\mathrm{GT}}_{i})^{N}},
\end{equation}
which is known to perform well even in the presence of unseen species~\cite{ChaoShen_2003}.

\begin{rem}
As noted in \cite{Hausser}, many other estimators of the Shannon entropy have been proposed.
However, most of these methods either assume that the total number of species is known or require numerical integration. 
In this study, we focus on methods without such restrictions.
Hence, we consider the maximum likelihood estimator, the Miller--Madow estimator, and the Chao--Shen estimator as baselines for comparison.  
\end{rem}

In addition to these estimators, we also include the Shannon entropy estimated by $\mathrm{DPYM}(\mathbf{y},0,1)$ and $\mathrm{DPYM}(\mathbf{y},1/2,0)$ as benchmarks in our simulation study.

\subsection{Simulation setup}
Let us describe simulation settings.
The dimension $K$ of the true probability vector $\bm{p}$ is fixed at $5000$, and the sample size $N$ is set to one of the following values: 10, 30, 50, 100, 200, 400, 600, 800, 1000, 3000, 5000, 8000, 10000, 20000.
The probability vector $\bm{p}$ is generated according to the following five scenarios:
\begin{enumerate}
    \item Sparse distribution, following a Dirichlet distribution with parameter $\alpha=0.1$.
    \item Homogeneous distribution, following a Dirichlet distribution with parameter $\alpha=1$.
    \item Nearly uniform distribution, following a Dirichlet distribution with parameter $\alpha=10$.
    \item Combined distribution, following Dirichlet distribution whose parameter vector has the first half components set to $0.1$ and the second half components set to $10$.
    \item Zipf distribution with exponent $s=1$, truncated at rank $5000$.
\end{enumerate}
For each scenario and each sample size, we perform $1000$ replications.
In each replication, the true probability vector $\bm{p}$ is independently regenerated.
The performance of the estimators is evaluated using the mean squared error (MSE).

\subsection{Summary of results from simulations}
\begin{figure}[tbp]
    \centering
    \begin{minipage}{0.45\textwidth}
        \centering
        \includegraphics[width=\textwidth]{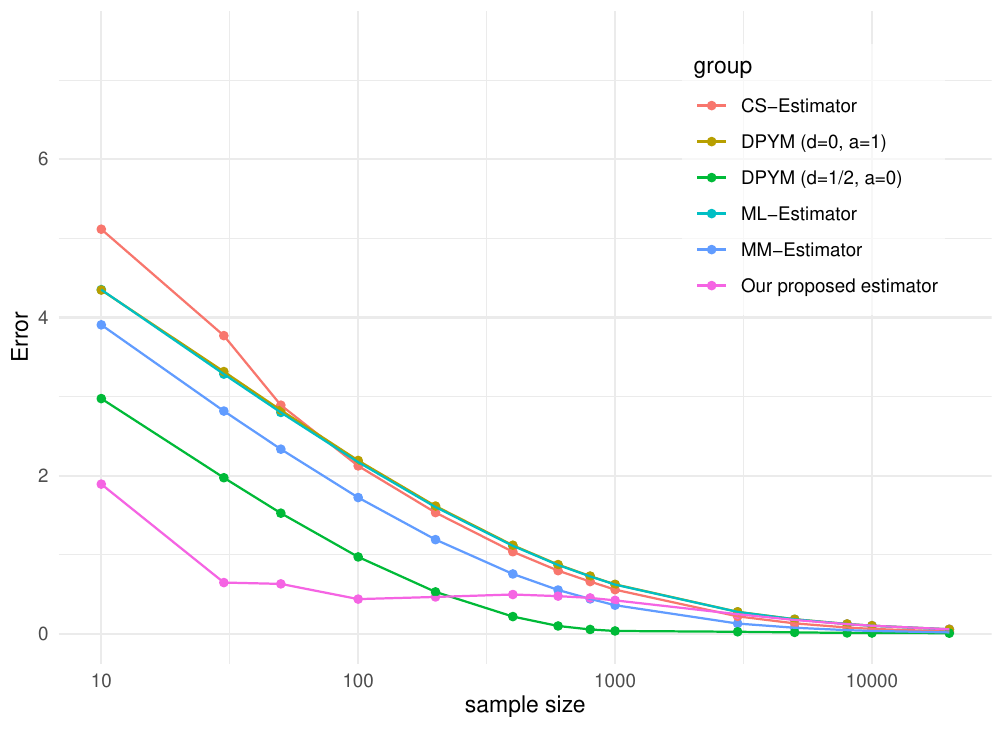}
    \end{minipage}%
    \hspace{1em}
    \begin{minipage}{0.45\textwidth}
        \centering
        \includegraphics[width=\textwidth]{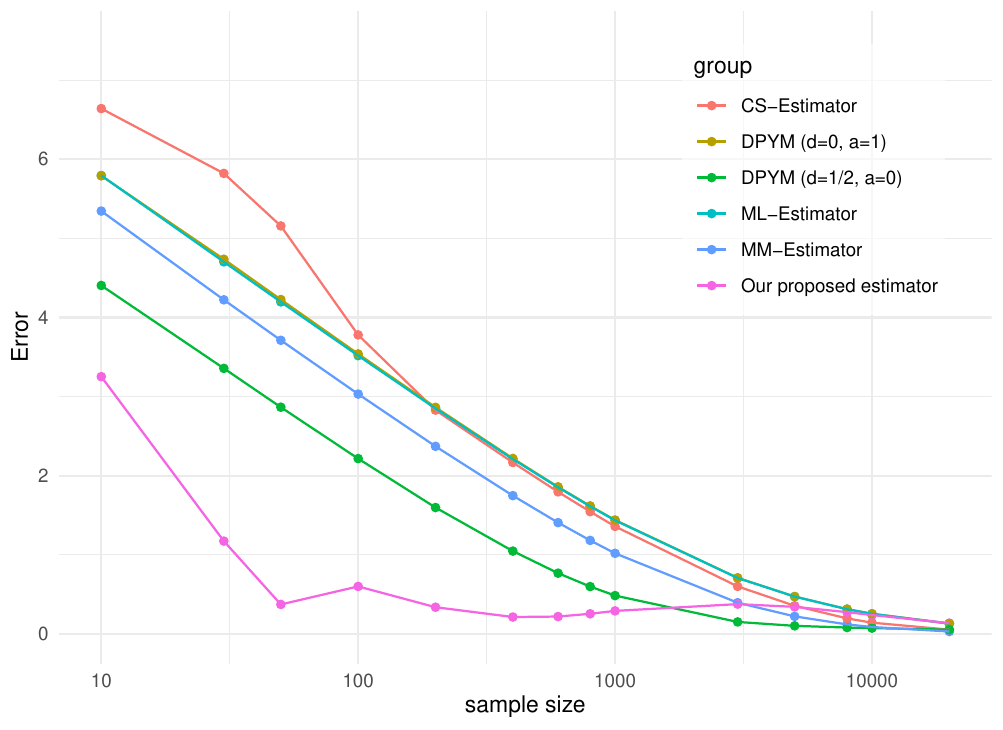}
    \end{minipage}%
    \vspace{1em}
    \begin{minipage}{0.45\textwidth}
        \centering
        \includegraphics[width=\textwidth]{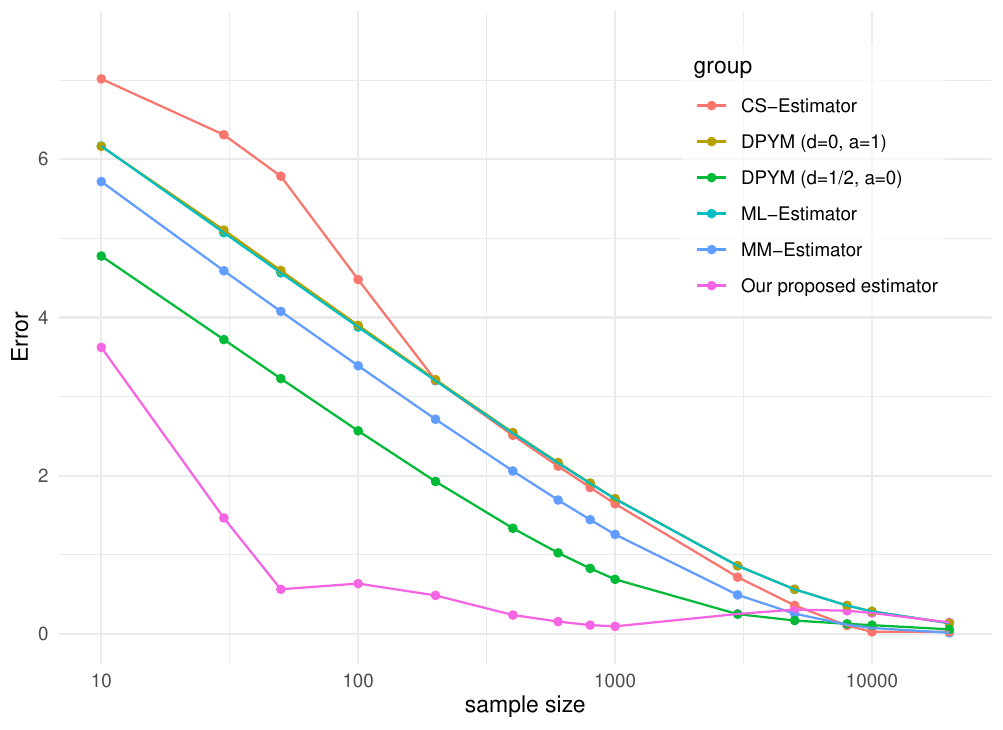}
    \end{minipage}%
    \hspace{1em}
    \begin{minipage}{0.45\textwidth}
        \centering
        \includegraphics[width=\textwidth]{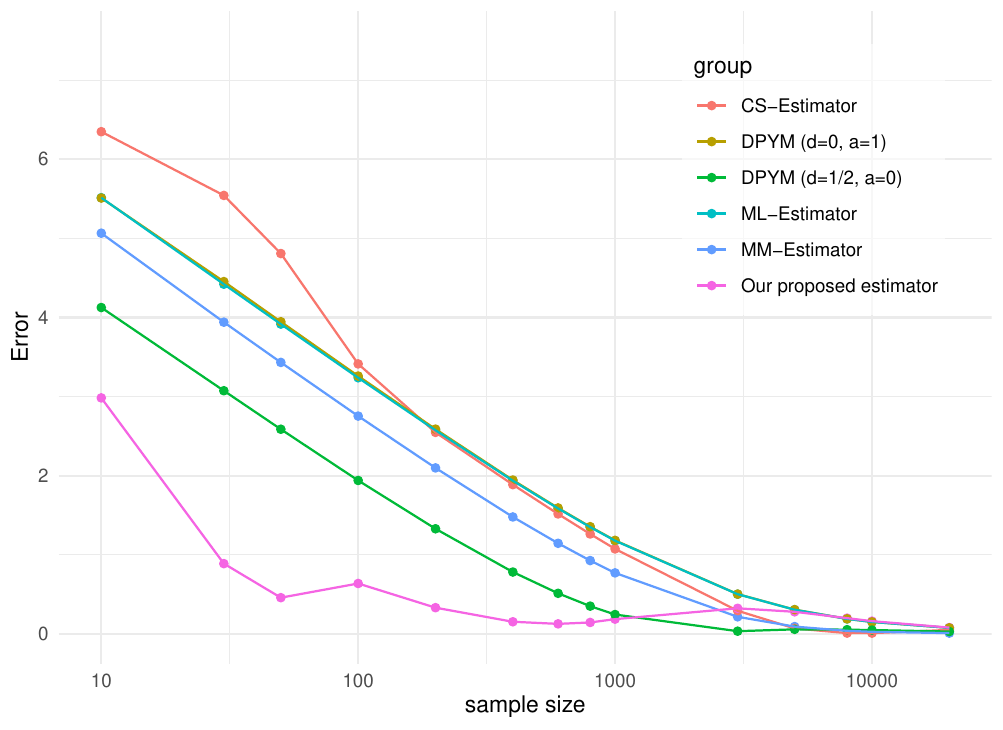}
    \end{minipage}%
    \vspace{1em}
    \begin{minipage}{0.45\textwidth}
        \centering
        \includegraphics[width=\textwidth]{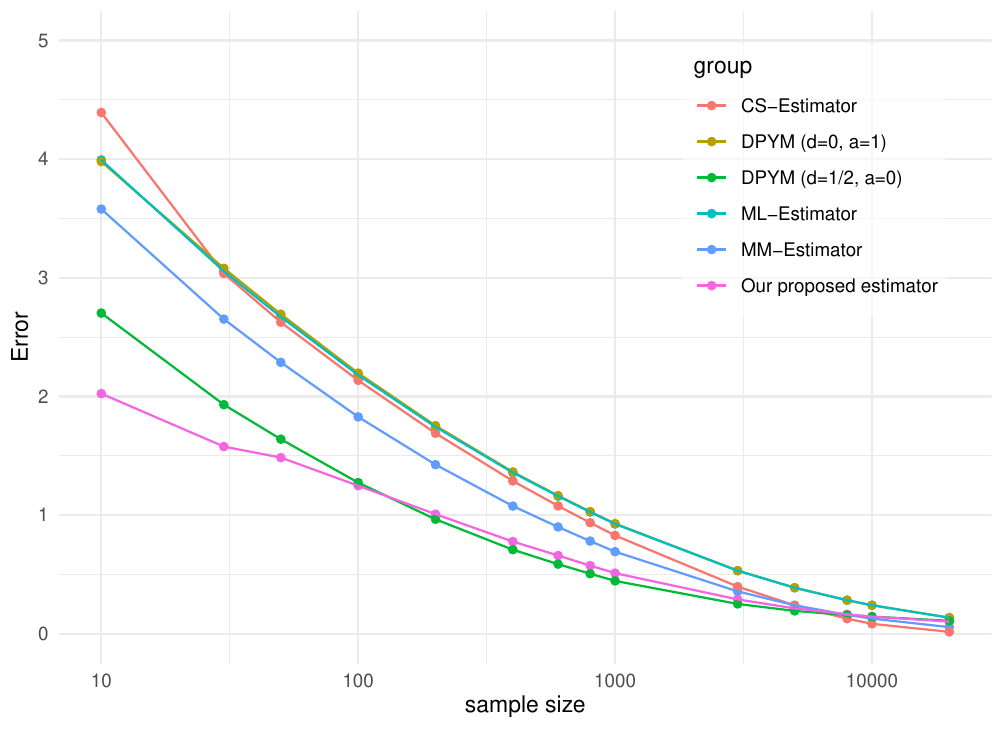}
    \end{minipage}%
    \caption{Comparing the four different entropy estimators (MLE, Miller--Madow, Chao--Shen and the proposed estimator) in five different sampling scenarios. The estimators are compared in terms of MSE of the estimated entropy. The dimension is fixed at $K = 5000$ while the sample size $N$ varies from 10 to 20000.}
    \label{fig:MSE}
\end{figure}

Figure~\ref{fig:MSE} displays the results of the simulation study, which can be summarized as follows:
\begin{itemize}
    \item When $N$ is extremely small relative to the dimension $K$ of $\bm{p}$, the proposed estimator consistently outperforms the other estimators across all scenarios.
    \item In the small-sample regime, $\mathrm{DPYM}(\mathbf{y},1/2,0)$ typically achieves the second-best performance after the proposed estimator, and attains the lowest MSE when the sample size is moderate to large.
    \item Even when $N$ is sufficiently large, the proposed estimator achieves comparable error to the other methods, demonstrating stable estimation.
    \item As the distribution becomes closer to uniform, estimation error of our proposed estimator tends to increase when $N$ is on the same order as the dimension $K$ of $\bm{p}$.
\end{itemize}

The proposed estimator performs well even with small sample sizes and thus achieves the objective of this study, namely stable and accurate estimation of the Shannon entropy in high-dimensional, small-sample settings.
For sufficiently large sample sizes, it retains accuracy comparable to existing estimators, confirming its reliability in large-sample regimes.
However, a reduction in accuracy is observed for moderate sample sizes.
This phenomenon appears to result from a misfit in parameter selection induced by the upper bound function.
Addressing this limitation constitutes an important direction for future theoretical and practical work.

\section{Application to real-world data}\label{sec:application}

In this section, we apply the proposed estimator to real-world data in order to evaluate diversity. We consider two datasets: tree species data from a tropical rainforest and word frequency data from news articles.
By comparing the results with conventional estimators such as the MLE and the Chao--Shen estimator, we highlight the advantages of our approach, particularly in situations with limited sample sizes or high underlying diversity.

\subsection{Tropical tree community data}
We use tree abundance data at a 1-hectare scale from Barro Colorado Island (BCI), Panama, to estimate entropy based on species occurrence frequencies within each plot. 
The aim of this analysis is to investigate differences in species diversity and biases in species distribution within a spatially structured forest. 

The Barro Colorado Island (BCI) dataset was provided by Condit et al.~\cite{Condit_2002} to illustrate analytical methods in the \texttt{vegan} package in \textsc{R}.
It comprises tree abundance data collected across 50 one-hectare plots in a tropical forest in Panama.
Within each plot, all trees with a diameter at breast height (DBH) of at least 10 cm were counted, and the abundances of 225 species were recorded.

The BCI dataset has been widely used for biodiversity evaluation, studies of species spatial distribution, and analysis of community structure.
In particular, it has been employed to estimate diversity indices such as the Shannon entropy and to analyze spatial patterns of species turnover (beta diversity).

Figure~\ref{fig:BCI} presents the estimated Shannon entropy for each plot, with plots on the horizontal axis and estimated values on the vertical axis. 
In the BCI dataset, species counts based on observed frequencies can be considered to roughly reflect true species richness. 
Accordingly, the proposed estimator produces values similar to those obtained with the MLE.
\begin{figure}[tbp]  
  \centering
  \includegraphics[width=0.6\textwidth]{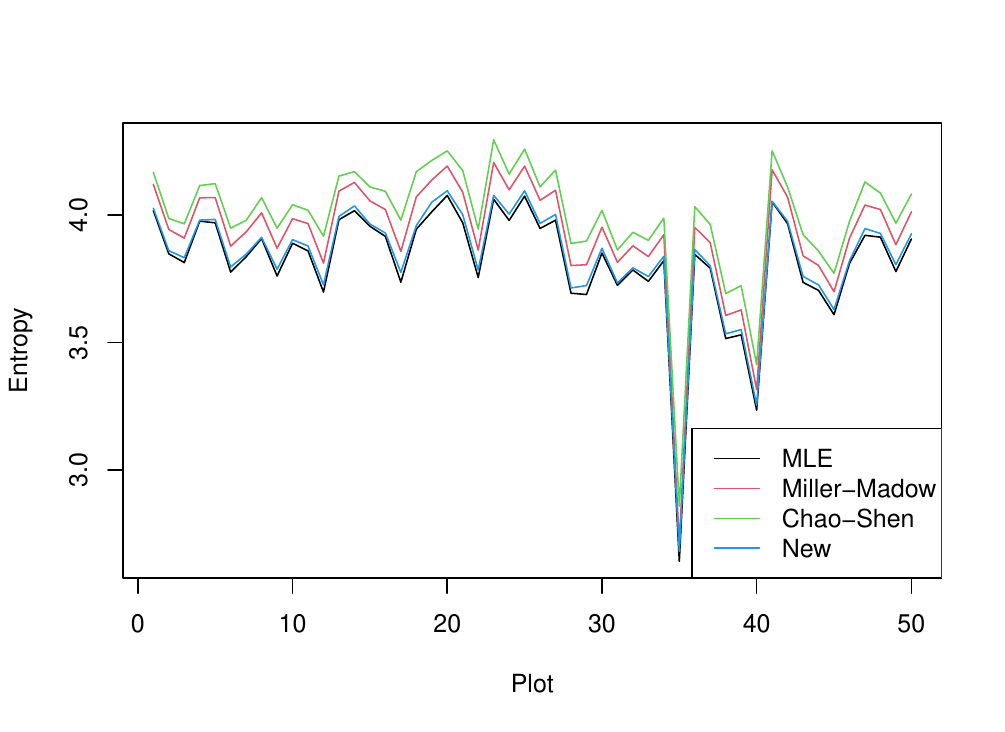} 
    \caption{Estimates of the Shannon entropy for each plot in the BCI dataset}
    \label{fig:BCI}
\end{figure}

\subsection{20 Newsgroups}

We estimate the entropy of word occurrence probabilities using the 20 Newsgroups dataset, a large-scale corpus for text classification.
In the field of natural language, the Shannon entropy is widely used to quantify the information content of words or characters.
Entropy serves as a measure of the freedom of word choice and unpredictability, and is commonly employed to analyze linguistic complexity and expressive power~\cite{Arora,Bentz_2017}.

The 20 Newsgroups dataset consists of approximately 18000 posts collected from 20 different newsgroups, with each document provided along with a topic label~\cite{Lang_1995}.
Text preprocessing was performed using the \texttt{tm} package in \textsc{R}.
Each document was loaded as a corpus, and standard preprocessing steps, such as removal of punctuation and invalid characters, were applied.
Based on the empirical distribution of word occurrences, we estimated the entropy of the word distribution for each document.

Figure~\ref{fig:News} shows the estimated entropy for each document, averaged within each category.
In this dataset, the true vocabulary size is expected to be very large.
Consistent with this expectation, the MLE and the Miller--Madow estimator tend to produce smaller entropy estimates, whereas the Chao--Shen estimator and the proposed estimator yield higher values, better capturing the underlying diversity of words.

\begin{figure}[tbp]  
  \centering
  \includegraphics[width=0.6\textwidth]{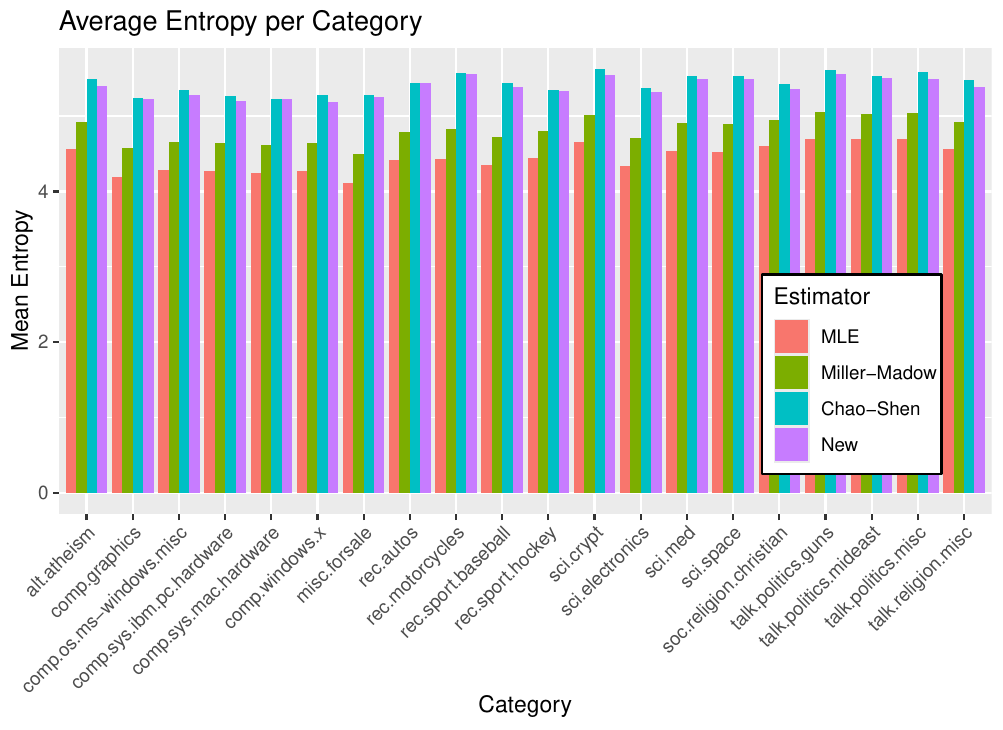} 
    \caption{Estimates of the Shannon entropy for each document in the 20 Newsgroups dataset}
    \label{fig:News}
\end{figure}

\section{Conclusion}\label{sec:conclusion}
In this study, we have proposed a method for accurately estimating the Shannon entropy of a population distribution even when the number of species is unknown and possibly very large.
The proposed estimator corrects the naive plug-in estimator by adding a correction term derived from a Pitman--Yor process.
Furthermore, we derived convergence conditions for the Shannon entropy in the general setting where discrete probability vectors follow a regularly varying function, and used these results to establish the consistency of the proposed estimator.
While the method is consistent for large sample sizes, numerical experiments demonstrated its particular usefulness in situations where the number of species exceeds the sample size, outperforming existing estimators.
Real-data applications suggest that the proposed estimator 
accommodate to different sampling regimes, exhibiting comparatively small bias correction for large samples and stronger bias correction for smaller ones. 

Several directions remain for future work.
First, a more detailed theoretical analysis under finite-sample conditions is warranted.
Although consistency was established in the large-sample regime, the convergence rate and finite-sample errors have yet to be fully characterized.
Second, while this study focused on the Shannon entropy, extensions to other information-theoretic measures, such as Rényi or Tsallis entropy, represent important avenues for further research.
Third, the correction term based on the Pitman--Yor process is flexible and could be adapted to incorporate different prior assumptions about the tail behavior or to suit specific application domains. 
Finally, potential extensions include applications to data with dependency structures or to large-scale datasets, which remain to be explored.

\section*{Acknowledgments}
The second author was supported in part by Japan Society for the Promotion of Science KAKENHI Grant Number 25K07133.


\begin{thebibliography}{99}
    \bibitem{Archer2014} Archer, E., Park, M., Pillow, J. W. (2014). ``Bayesian Entropy Estimation for Countable Discrete Distributions'' \textit{J. Mach. Learn. Res.} \textbf{15} 2833--2868
    \bibitem{Arora} Arora, A., Meister, C., Cotterell, R. (2022). ``Estimating the Entropy of Linguistic Distributions'' \textit{Proceedings of the 60th Annual Meeting of the Association for Computational Linguistics} \textbf{2}, 175--195
    \bibitem{Baccetti_2013}Baccetti, V., Visser, M. (2013). ``Infinite Shannon entropy'' \textit{J. Stat. Mech. Theory Exp.} no. 4, P04010, 12 pp.
    \bibitem{Bentz_2017}Bentz, C., Alikaniotis, D., Cysouw, M., Ferrer-i-Cancho, R. (2017). ``The Entropy of Words—Learnability and Expressivity across More than 1000 Languages''. \textit{Entropy} \textbf{19}, no.275. 
    \bibitem{Bingham_1987}Bingham, N. H., Goldie, C. M., Teugels, J.L. (1987). \textit{Regular Variation}. Cambridge University Press
    \bibitem{ChaoLee_1992} Chao, A., Lee, S. M. (1992). ``Estimating the Number of Classes via Sample Coverage'' \textit{J. Amer. Statist. Assoc.} \textbf{87}, no.417, 210--217
    \bibitem{ChaoShen_2003} Chao, A., Shen, T. J. (2003). ``Nonparametric estimation of Shannon's index of diversity when there are unseen species in sample''  \textit{Environ. Ecol. Stat.} \textbf{10}, 429--443
    \bibitem{Condit_2002} Condit, R., Pitman, N., Leigh, E. G., Chave, J., Terborgh, J., Foster, R. B., Nuñez, P., Aguilar, S., Valencia, R., Villa, G., Muller-Landau, H. C., Losos, E., Hubbell S. P. (2002).  ``Beta-diversity in tropical forest trees'' \textit{Science} \textbf{295}, 666--669.
    \bibitem{Esty_1983} Esty, W.W. (1983).
    ``A normal limit law for a nonparametric estimator of the coverage of a random sample''
    \textit{Ann. Statist.} \textbf{11}, no. 3, 905--912.
    \bibitem{Good_1953} Good, I. J. (1953). 
    ``The Population Frequencies of Species and the Estimation of Population Parameters'' 
    \textit{Biometrika} \textbf{40}, 237--264
    \bibitem{Ghosal_2017}Ghosal, S., van der Vaart A. (2017) \textit{Fundamentals of Nonparametric Bayesian Inference}, Cambridge University Press.
    \bibitem{Hausser} Hausser, J., Strimmer, K. (2009). 
    ``Entropy Inference and the James-Stein Estimator, with Application to Nonlinear Gene Association Networks''
    \textit{J. Mach. Learn. Res.} \textbf{10}, 1469--1484
    \bibitem{Jiao_2015} Jiao, J., Venkat, K., Han, Y., Weissman, T. (2015) ``Minimax Estimation of Functionals of Discrete Distributions,'' \textit{IEEE Trans. Inform. Theory} \textbf{61} no. 5, 2835--2885
    \bibitem{Jonson_2005} Johnson, N. L., Kotz, S., Kemp, A. W. (2005) \textit{Univariate Discrete Distributions Third Edition} Wiley Series in Probability and Statistics.  John Wiley \& Sons, Ltd.
    \bibitem{Lang_1995} Lang, K. (1995). ``NewsWeeder: Learning to Filter Netnews''. \textit{ Machine learning proceedings 1995}. 331--339.
    \bibitem{Leinster} Leinster, T. (2021).
    \textit{Entropy and Diversity. the Axiomatic Approach}
    Cambridge University Press, Cambridge.
    \bibitem{McAllester_2000} McAllester, D., Schapire, R. E. (2000). ``On the Convergence Rate of Good-Turing Estimators'' \textit{COLT} 1--6
    \bibitem{Miller_1955} Miller, G. A. (1955) ``Note on the bias of information estimates'' In H. Quastler, editor, \textit{Information Theory in Psychology II-B}, 95--100. Free Press, Glencoe, IL.
    \bibitem{Nair_2022}Nair, J., Wierman, A., Zwart, B. (2022). \textit{The Fundamentals of Heavy Tails Properties, Emergence, and Estimation}, Cambridge University Press.
    \bibitem{Paninski_2004} Paninski, L. (2004). ``Estimating Entropy on $m$ Bins Given Fewer Than   $m$ Samples'' \textit{IEEE Trans. Inform. Theory} \textbf{50}, no. 9, 2200--2203
    \bibitem{Patil_1982} Patil, G.P., Taillie, C. (1982). ``Diversity as a concept and its measurement'' \textit{J. Amer. Statist. Assoc.} \textbf{77}, no.379, 548--561.
    \bibitem{Pitman_1997} Pitman, J., Yor, M. (1997). ``The Two-Parameter Poisson-Dirichlet Distribution Derived from a Stable Subordinator''  \textit{Ann. Probab.} \textbf{25}, no. 2, 855--900
    \bibitem{Resnick_2007} Resnick, S. I. (2007).
    \textit{Heavy-Tail Phenomena: Probabilistic and Statistical Modeling}, Springer New York.
    \bibitem{Yihong_2016} Wu, Y., Yang, P. (2016).
    ``Minimax Rates of Entropy Estimation on Large Alphabets via Best Polynomial Approximation'' \textit{IEEE Trans. Inform. Theory} \textbf{62}, no. 6, 3702--3720
    \bibitem{Zhang_2009} Zhang, C.-H., Zhang, Z. (2009).
    ``Asymptotic normality of a nonparametric estimator of sample coverage''
    \textit{Ann. Statist.} \textbf{37}, no. 5A, 2582--2595.
\end{thebibliography}
\end{document}